%% file: main.tex
\documentclass[envcountsame,envcountsect]{llncs}
\usepackage{amssymb}
\usepackage{amsmath}
\usepackage{tikz}
\usetikzlibrary{topaths,calc}
\usepackage{pgfplots}
\pgfplotsset{compat=1.17}

\RequirePackage{graphicx}
\RequirePackage{algorithm}
\RequirePackage[endLComment=]{algpseudocodex}
\usepackage{mathtools}
\usepackage{subcaption}
\usepackage[nameinlink,capitalize]{cleveref}
\DeclareMathOperator*{\argmax}{arg\,max}
\DeclareMathOperator*{\argmin}{arg\,min}

\newcommand{\baruch}[1]{{\color{blue} (Baruch :#1)}}

\newcommand{\negA}{\vspace{-0.05in}}
\newcommand{\negB}{\vspace{-0.1in}}
\newcommand{\negC}{\vspace{-0.18in}}

\newcommand{\mysection}[1]{\negC\section{#1}\negA}
\newcommand{\mysubsection}[1]{\negB\subsection{#1}\negA}

\newcommand{\myparagraph}[1]{\par\smallskip\par\noindent{\bf{}#1:~}}
\newcommand{\comment}[1]{}
\newcommand{\lrp}[1]{\left( #1 \right)}

\newcommand{\lrc}[1]{\left\{ #1 \right\}}

\DeclarePairedDelimiter{\floor}{\lfloor}{\rfloor}

\newcommand{\cP}{{\mathcal{P}}}
\newcommand{\cA}{{\mathcal{A}}}
\newcommand{\cB}{{\mathcal{B}}}

\newcommand{\cI}{{\mathcal{I}}}

\newcommand{\cj}{{\mathcal{J}}}
\newcommand{\cH}{{\mathcal{H}}}

\newcommand{\cV}{{\mathcal{V}}}
\newcommand{\cE}{{\mathcal{E}}}

\newcommand{\OPT}{\textnormal{OPT}}

\newcommand{\eps}{{\varepsilon}}

\newcommand{\N}{\mathbb{N}}

\newcommand{\RCP}{{\sc rcp}}
\newcommand{\URCP}{{\sc u-rcp}}
\newcommand{\CT}{{\sc ct}}
\newcommand{\CPO}{{\sc cpo}}
\newcommand{\BPCC}{{\sc bpcc}}
\newcommand{\DKSH}{{\sc d}\textit{k}{\sc sh}}
\newcommand{\zzp}{{\mathbb{Z}_{\geq 0}}}
\newcommand{\zzpos}{{\mathbb{Z}_0^+}}
\newcommand{\anc}[2]{\ensuremath{\textsc{Ancs}_{{#1}}({#2})}}
\newcommand{\des}[2]{\ensuremath{\textsc{Desc}_{{#1}}({#2})}}

\newcommand{\AC}{\ensuremath{\textsc{AC}}}
\newcommand{\AD}{\ensuremath{\textsc{AD}}}
\newcommand{\LB}{\ensuremath{\textsc{LB}}}

\newcommand{\topA}{\ensuremath{\textsc{top}A}}

\newcommand{\na}{\ensuremath{\textsc{sa}}}
\newcommand{\lo}{\ensuremath{\textsc{lo}}}
\newtheorem{myclaim}[lemma]{Claim}
\newtheorem{observation}[lemma]{observation}

\def\sm{\setminus}
\def\sse{\subseteq}
\def\ss{\subset}

\def\calC{\mathcal{C}}

\def\calP{\mathcal{P}}
\def\calQ{\mathcal{Q}}

\def\calU{\mathcal{U}}

\begin{document}
	
\def\claimproof{\proof}
\def\endclaimproof{\hfill$\square$\\}

\pagestyle{plain}

\title{
Approximations and Hardness of Covering and 
Packing Partially Ordered Items
}


\comment{
\author{Ilan Doron-Arad}\affil{Computer Science Department, Technion, Haifa, Israel. \texttt{\{idoron-arad,naor,hadas\}@cs.technion.ac.il}}
\author{Guy Kortsarz}\affil{Computer Science Department, Rutgers University-Camden, Camden, NJ, USA. \texttt{guyk@camden.rutgers.edu}}
\author{Seffi Naor}
\author{Baruch Schieber}\affil{Computer Science Department, New Jersey Institute of Technology, Newark, NJ, USA. \texttt{sbar@njit.edu}}
\author{Hadas Shachnai}
}

\sloppy
\setcounter{page}{1} 
\title{Approximations and Hardness of Covering and
	Packing Partially Ordered Items}
\author{Ilan Doron-Arad\inst{1} \and Guy Kortsarz\inst{2} \and Joseph (Seffi) Naor\inst{1} \and \\ Baruch Schieber\inst{3} \and Hadas Shachnai\inst{1}}
\institute{Computer Science Department, Technion, Haifa, Israel. \email{\{idoron-arad,naor,hadas\}@cs.technion.ac.il.}
\and 
Computer Science Department, Rutgers University-Camden, Camden, NJ, USA.
\email{guyk@camden.rutgers.edu.}
\and
Computer Science Department, New Jersey Institute of Technology, Newark, NJ, USA. \email{sbar@njit.edu.}
}


\maketitle

\input{abstract}

\input{intro}

\input{algorithm}

\input{open_probs}

\bibliographystyle{splncs04}
\bibliography{bibfile}

\appendix

\input{motivation}

\input{hardness}

\input{reduction}

\input{uniform}

\input{bounded}

\input{missingProofs}

\comment{

\input{abstract}
\input{intro}
\input{prelims}

\input{algorithm}
\input{hardness}
\input{reduction}
\input{uniform}
\input{bounded}
}

\end{document}

%% file: abstract.tex
\begin{abstract}

Motivated by applications in production planning and storage allocation in hierarchical databases, we initiate the study of {\em covering partially ordered items} ({\CPO}). Given
a value $k \in \N$, and 
a directed graph $G=(V,E)$ where each vertex has a size in $\{0,1, \ldots , k\}$, 
we seek a collection of subsets of vertices $C_1, \ldots, C_t$ that {\em cover} all the vertices, such that for any $1 \leq j \leq t$, the total size of vertices in $C_j$ is bounded by $k$, and there are no edges from $V \setminus C_j$ to $C_j$. 
The objective is to minimize the number of subsets $t$. {\CPO} is closely related to the {\em rule caching problem} ({\RCP}) that has been widely studied 
in the networking area. The input for {\RCP} is a directed graph $G=(V,E)$, a profit function $p:V \rightarrow \mathbb{Z}_{0}^+$, and $k \in \N$.
The output is a subset $S \subseteq V$ of maximum profit such that $|S| \leq k$ and there are no edges from $V \setminus S$ to $S$. 

Our main result is a $2$-approximation algorithm for {\CPO} on out-trees, complemented by an asymptotic $1.5$-hardness of approximation result.	
We also give a two-way reduction between {\RCP} and the {\em densest $k$-subhypergraph} problem, surprisingly showing that the problems are equivalent w.r.t. polynomial-time approximation within any factor $\rho \geq 1$. This implies that {\RCP} cannot be approximated within factor $|V|^{1-\eps}$ for any fixed $\eps>0$, under standard
complexity assumptions. Prior to this work, {\RCP} was just known to be strongly NP-hard.
We further show that there is no EPTAS for the special case of {\RCP} where the profits are uniform, assuming Gap-ETH. Since this variant admits a PTAS, we essentially resolve the complexity status of this problem.
\end{abstract}

%% file: intro.tex
\mysection{Introduction}

Partially ordered entities are ubiquitous in the mathematical modeling of scheduling problems, distributed storage allocation, production planning, and unified language models. Often, the partial order represents either
precedence constraints or dependencies among entities (or items). 
Motivated by applications in production planning~\cite{WZZZG22,BSH11} and distributed storage allocation in hierarchical databases~\cite{YJDYD09},
we introduce the {\em covering partially ordered items} (\CPO) problem. 
An instance of {\CPO} consists of a directed graph $G=(V,E)$, a value $k \in \mathbb{N}$, and a size function $s:V \rightarrow [0:k]$.\footnote{For $i,j \in {\zzpos}$, we denote by $[i:j]$ the set of integers $\{ i,i+1, \ldots , j\}$.} A {\em configuration} is a subset of vertices $U \subseteq V$ such that $s(U) \leq k$,\footnote{For a set $A$, a function $f:A \rightarrow \mathbb{X}$, and $B \subseteq A$, define $f(B) = \sum_{b \in B} f(b)$.} 
and $U$ is closed under precedence constraints; that is, for any $u \in U$ and $(z,u) \in E$ it holds that $z \in U$. A {\em feasible solution} is a set of configurations $C_1,\ldots, C_t$ that {\em covers} $V$, namely $\bigcup_{j\in [1:t]} C_j = V$. 
The {\em cardinality} of the solution is $t$, the number of configurations. The goal is to find a feasible solution of minimum cardinality.

{\sc cpo} can be applied to optimize the distributed storage of large hierarchical data in {\em unified medical language systems} (UMLS)~\cite{YJDYD09}. UMLS data is 
often distributed over several databases of bounded size. Due to the hierarchical nature of the medical taxonomy, each database needs to be closed under this hierarchy relation. The problem of minimizing the number of distributed databases of the UMLS data translates to a {\sc cpo} problem instance.

Another application of {\sc cpo} arises in production planning for steel mills that employ continuous casting~\cite{WZZZG22,BSH11}. The steel-making process has high energy consumption. One way to save energy is by employing continuous casting and direct charging. In this routine, the molten steel is solidified into slabs and rolled into finished products of various sizes continuously, with no need to reheat the steel in the process. Each finished product requires specific casting, rolling, and thermal treatments in a given order, which can be modeled by a {\em directed acyclic graph} (DAG). A main challenge is to assign the finished products to batches whose size is dictated by the size of the ladle furnace while minimizing the amount of repeated operations. This gives rise to an instance of {\sc cpo}. 

A natural greedy approach for solving {\sc cpo} is to repeatedly find, among all subsets of vertices that can be feasibly assigned to a single configuration, a subset that maximizes the size of yet unassigned vertices.
This single configuration problem is a variant of
the well known {\em rule caching problem} ({\sc rcp}) that has been studied extensively  \cite{dong2015rule,yan2014cab,sheu2016wildcard,huang2015cost,li2019tale,stonebraker1990rules,li2015fdrc,gao2021ovs,cheng2018switch,rottenstreich2016optimal,li2020taming,gamage2012high,yan2018adaptive,rottenstreich2020cooperative,rastegar2020rule,yang2020pipecache}. 
An instance of {\sc rcp} consists of a directed graph $G=(V,E)$, a profit function $p:V \rightarrow {\zzpos}$, and a value $k \in  \mathbb{N}$. We seek a subset of vertices $U \subseteq V$ which is closed under precedence constraints, such that $|U| \leq k$, and $p(U) = \sum_{u \in U} p(u)$ is maximized. In Appendix~\ref{sec:motivation} we describe central applications of {\RCP} in networking and the blockchain technology.

Prior to this work {\sc rcp} was just known to be strongly NP-hard~\cite{borradaile2012knapsack,pferschy2018improved}. Our initial attempt towards solving {\sc cpo} was to find a good approximation for {\sc rcp}. 
Surprisingly, we were able to show an equivalence between {\sc rcp} and the 
{\em densest $k$-subhypergraph} ({\DKSH}) problem w.r.t. approximability. The input for {\DKSH} consists of a hypergraph $G=(V,E)$ and a value $k \in {\N}$. The goal is to find a subset of vertices $S \subseteq V$ of cardinality $k$ that maximizes the number (or weight) of induced hyperedges (a more formal definition is given in \Cref{sec:2}).\footnote{{\DKSH} has also been widely studied (see, e.g.,~\cite{denseksub,hajiaghayi2006minimum,unioneden} and the references therein).} 

Unfortunately, {\DKSH} is known to be hard to approximate within a factor of $|V|^{1-\eps}$, for $\eps \in (0,1)$, assuming the Small Set Expansion Conjecture (by combining the results of \cite{manurangsi2017inapproximability} and \cite{hajiaghayi2006minimum}).  
This implies the same hardness of approximation for {\sc rcp} (see Section~\ref{sec:results}).
Given this hardness result, we expect {\sc cpo} to be hard to approximate on general graphs. Thus, we consider the special case of {\sc cpo} where $G$ is an out-tree. We call this problem {\em covering partially ordered items on out-trees ({\sc ct})}.
To the best of our knowledge, {\CT} is studied here for the first time.
We note that when $G$ is an in-tree {\sc cpo} is trivial since
the problem has a feasible (and unique) solution {\em iff} the total size of the vertices is at most $k$.

\mysubsection{Our results}
\label{sec:results}
Our first result is an approximation algorithm for {\CT}.
Recall that, for $\alpha \geq 1$, $\cA$ is an $\alpha$-approximation algorithm for a minimization (maximization) problem $\Pi$ if, for any instance of $\Pi$, the output of $\cA$ is at most $\alpha$ (at least $1/\alpha$) times the optimum.
\begin{theorem}
	\label{thm:main}
	There is a polynomial time $2$-approximation algorithm for \textnormal{\sc ct}. 
\end{theorem}

While out-trees have a simple structure, allowing for a greedy-based bottom-up approach in solving {\sc ct}, the analysis of our approximation algorithm is nontrivial and requires extra care to make sure that the approximation bound has no additive terms (see below).

{\sc ct} generalizes the classical {\em bin packing} ({\sc bp}) problem. The input for {\sc bp} is a set of items and a value $k \in \mathbb{N}$. Each item has a size in $[0:k]$, and the goal is to assign the items into a minimum number of bins of capacity $k$.\footnote{We use the definition of {\sc bp} as given in~\cite{GJ79}. In an alternative definition found in the literature, bin capacities are normalized to one, and item sizes are in $[0,1]$.} An instance of {\sc bp} is reduced to an instance of {\sc ct} on a star graph by generating a leaf for each item of the {\sc bp} instance and adding a root vertex of size zero. This trivial reduction implies that {\sc ct} is strongly NP-hard~\cite{GJ79}. 

Interestingly, we show that in contrast to {\sc bp}, {\sc ct} does not admit an {\em asymptotic polynomial-time approximation scheme} (APTAS), or even an asymptotic approximation strictly better than $\frac{3}{2}$. This separates {\sc ct} from {\sc bp} which admits also an additive logarithmic approximation \cite{hoberg2017logarithmic}.

\begin{theorem}
	\label{thm:hard}
	For any $\alpha < \frac{3}{2}$, there is no asymptotic $\alpha$-approximation for \textnormal{\sc ct} unless \textnormal{P=NP}.   
\end{theorem}
\noindent
Next, we study the hardness of \RCP.

\begin{theorem}
	\label{thm:EQ}
	For any $\rho \geq 1$, there is a  $\rho$-approximation for \textnormal{\sc rcp} if and only if there is a  $\rho$-approximation for \textnormal{\DKSH}. 
\end{theorem}

\begin{corollary}
	\label{cor:RCP}
	Assuming the Small Set Expansion Hypothesis (SSEH) and  \textnormal{NP} $\neq$ \textnormal{BPP}, for any $\eps > 0$ there is no $|V|^{1-\eps}$-approximation for \textnormal{\sc rcp}. 
\end{corollary}

We give a tight lower bound also for the previously studied special case of
{\em uniform} {\sc rcp (u-rcp)}~\cite{borradaile2012knapsack,bonsma2010most}.\footnote{In~\cite{borradaile2012knapsack}, {\sc u-rcp} is called uniform directed all-neighbor knapsack problem.} 
In {\sc u-rcp} the vertices have uniform (unit) profits (i.e., $p(v) = 1~\forall v \in V$). 
While {\sc u-rcp} is known to admit a PTAS~\cite{borradaile2012knapsack,bonsma2010most}, the question of whether the problem admits an EPTAS or even an FPTAS remained open.\footnote{We give formal defintions relating to approximation schemes in Appendix~\ref{sec:uniform}.} Our next result gives a negative answer to both of these questions, posed in~\cite{borradaile2012knapsack,bonsma2010most}. 

\begin{theorem}
	\label{thm:2}
	Assuming \textnormal{Gap-ETH}, there is no \textnormal{EPTAS} for \textnormal{\sc u-rcp}. 
\end{theorem}

Finally, we show that {\RCP} remains essentially just as hard when the in-degrees and 
out-degrees are bounded.

\begin{theorem}
    \label{thm:2thm}
    A $\rho$-approximation algorithm for \textnormal{\sc rcp} instances with in-degrees and out-degrees bounded by $2$, for any $\rho \geq 1$, implies a $\rho$-approximation for \textnormal{\sc rcp}. 
\end{theorem}

Due to space constraints, we include in the paper body only the proof of  Theorem~\ref{thm:main} and defer the proofs of the other theorems to the Appendix.

\myparagraph{Techniques}
Our algorithm for {\CT} covers the vertices in a given out-tree $T$ in a bottom-up fashion, starting from the leaves. The key players in this process are vertices called {\em anchors} which define the candidate subtrees for covering in each iteration.  Interestingly, we show that the subtree associated with a specific anchor $a$ (including also all of $a$'s ancestors) can be covered efficiently by using the naive NextFit algorithm.

To eliminate {\em additive} terms in the approximation guarantee (i.e., obtain an {\em absolute} ratio of $2$), a crucial step in the algorithm is to distinguish
in each call to NextFit between the case where NextFit outputs an {\em even} vs. {\em odd} number of configurations. In the latter case, we discard the last configuration and cover the corresponding {\em leftover} vertices in a later iteration of the algorithm.

The crux of the analysis is to charge the number of subsets (i.e., configurations) used by the algorithm separately to each anchor. Consider a subtree of an anchor $a$, of total size $\na(a)$, covered at some iteration. Observing that each subset including vertices in this subtree must include also all the ancestors of $a$, we are able to show that the total number of subsets used is at most twice $\floor*{\frac{\na(a)}{k-h(a)+1}}$, where $h(a)$ is the total size of the ancestors of $a$ in $T$ (including $a$). To complete the analysis, we lower bound the number of subsets used in {\em any} feasible solution. This is done via an intricate calculation 
bounding the {\em number of occurrences} of each vertex 
$v$ in the subtree of an anchor $a$ in any feasible cover, which is the 
heart of the analysis. Our Greedy approach may be useful for other {\CPO} classes of instances in which the input graph $G$ has a {\em tree-like} structure (e.g., graphs of bounded treewidth).

Our proofs of hardness for {\CT} and {\RCP} use sophisticated constructions, most notably, to show a two-way reduction between {\RCP} and {\DKSH} (in Appendices~\ref{sec:2} and~\ref{sec:6}) and the hardness of {\RCP} with bounded degrees (see Appendix~\ref{sec:inOut}).

\myparagraph{Organization}
Section~\ref{sec:alg_ct} presents our approximation algorithm for {\sc ct} and the proof of Theorem~\ref{thm:main}, and Section~\ref{sec:open_probs} includes some open problems. In Appendix~\ref{sec:motivation} we describe common applications of {\RCP}, and Appendix~\ref{sec:hard} gives the hardness result for {\sc ct} (proof of Theorem~\ref{thm:hard}). Appendices~\ref{sec:2} and~\ref{sec:6} show the equivalence between {\sc rcp} and {\DKSH} (proofs of Theorem~\ref{thm:EQ} and Corollary~\ref{cor:RCP}). Appendix~\ref{sec:uniform} shows that there is no EPTAS for {\URCP} (Theorem~\ref{thm:2}), and Appendix~\ref{sec:inOut} proves the hardness of {\sc rcp} on graphs of in-degrees and out-degrees bounded by $2$  (Theorem~\ref{thm:2thm}). Finally, some missing proofs are given in Appendix~\ref{app:proofs}.

%% file: algorithm.tex
\mysection{Approximation Algorithm for {\CT}}
\label{sec:alg_ct}
In this section, we present our approximation algorithm for {\CT}. We start with some definitions and notations.
Let $T=(V, E)$ be an out-tree rooted at a vertex $r\in V$. Recall that in an out-tree all edges are oriented outwards from $r$. Thus, for an edge $(u,v)\in E$, vertex $u$ precedes $v$ on the (unique) path from $r$ to $v$. We say that $u$ is the \textit{parent} of $v$ and $v$ is a \textit{child} of $u$. 
More generally, if $u$ is on the (unique) path from $r$ to $v$ then $u$ is an \textit{ancestor} of $v$ and $v$ is a \textit{descendant} of $u$. A vertex $v$ is considered an ancestor of itself but \textit{not} a descendant of itself. Define $h(v)$ to be the total size of the vertices on the path from $r$ to $v$, which equals the total size of the ancestors of $v$. 

For $U \sse V$, let 
$T[U]$ be the subgraph of $T$ induced by $U$. If $T[U]$ is connected, then we say that $T[U]$ is a subtree of $T$. Note that in this case $T[U]$ is also an out-tree. From now on, we consider only induced subgraphs that are connected, namely subtrees of $T$.
If $r \in U$, then $T[U]$ is a subtree of $T$ rooted at $r$.

For an out-tree $T=(V, E)$ and a subset of vertices $U\sse V$, let $\anc{T}{U}$ be the set of the ancestors in $T$ of the vertices in $U$, and let $\des{T}{U}$ be the set of the descendants in $T$ of the vertices in $U$. Note that if $T[U]$ is a subtree of $T$ rooted at $r$, then $\anc{T}{U}=U$. In case $U$ is a singleton set, we omit the set notation; that is, for $v\in V$, let $\anc{T}{v}$ be the set of the ancestors of $v$ in $T$, and let $\des{T}{v}$ be the set of the descendants of $v$ in $T$.

We note that
if there is a vertex $v\in V$ for which $h(v)> k$ then there is no feasible solution. Also, if there is a leaf $\ell$ of $T$ for which $h(\ell) = k$ then any solution must include the set $\anc{T}{\ell}$ (of size $k$), and after adding this set to the solution, we can remove $\ell$ and all of its ancestors which are not ancestors of any other leaf. Thus, w.l.o.g. we assume that for any vertex $v\in V$ it holds that $h(v)< k$. Also, we note that if there is a leaf $\ell$ of $T$ of size $s(\ell)=0$ then we can remove $\ell$, solve for the resulting tree and then add $\ell$ to a subset in the cover that includes the parent of $\ell$ in $T$. Thus, w.l.o.g. we assume that for any leaf $\ell$ of $T$, $s(\ell)>0$.

The algorithm for computing a cover is iterative. In each iteration, we compute a partial cover as described below. We then continue to the next iteration with the subtree rooted at $r$ induced by the uncovered vertices and their ancestors. The algorithm terminates when either the set of uncovered vertices is empty or the total size of the vertices of the remaining subtree (rooted at $r$) is at most $k$, in which case these vertices form the last set in the cover.

In each iteration $t$ of the algorithm, we compute a subset of vertices $A_t\ss V$ that we call \textit{anchors}. We then compute a cover of {\em some} (potentially all) descendants of the anchors in $A_t$, and proceed to the next iteration.

\cref{alg:anchor} is the pseudo code of the iterative algorithm. Initially, $V_1=V$. Consider the $t$-th iteration, for $t\ge 1$. If $s(V_t)\le k$ then the algorithm terminates. Otherwise, define $A_t$ as the set of all the vertices $v\in V_t$ such that (i) the total size of the descendants of $v$ in $T[V_t]$ is more than $k-h(v)$, and (ii) the total size of the descendants of every child $u$ of $v$ in $T[V_t]$ is at most $k-h(u)=k-h(v)-s(u)$. 

Procedure NextFit given in \cref{alg:nextfit} is called for every $a\in A_t$. The input to Procedure NextFit is the tree $T_a$ defined as the rooted subtree that consists of the path from $r$ to $a$ and the descendants of $a$ in the subtree $T[V_{t(a)}]$ (see Figure~\ref{fig:Ta}).
When called for an anchor $a$, Procedure NextFit (\cref{alg:nextfit}) computes a cover of some (potentially all) descendants of $a$. The number of sets returned in this procedure call is even, and the
total size of the descendants of $a$ that are not covered by the sets returned by Procedure NextFit is at most $k-h(a)$. Let $U_t \sse V_t$ be the set of all descendants of anchors in $A_t$ that were covered in iteration $t$, together with all their ancestors. If $V_t= U_t$, then the algorithm terminates. Otherwise, we let $V_{t+1}$ be the set of ancestors of the vertices $V_t\sm U_t$ in $T[V_t]$ and continue to iteration $t+1$.
\negA
\begin{algorithm}[!ht]
	\caption{Feasible cover computation} \label{alg:anchor}
	\begin{algorithmic}[1]
		\Statex \emph{Input}: An out-tree $T=(V,E)$ rooted at $r$ and an integer $k>0$
		\Statex \emph{Output}: A feasible cover $\calC = C_1,\ldots,C_c$
		\State $V_1 \gets V$
		\State $\calC \gets \emptyset$
		\State $t \gets 1$
		\While{$s(V_t) > k$}
    		\State{$X_t\gets \lrc{u\in V_t\,|\, s(\des{T[V_t]}{u}) \le k-h(u)}$}
    		\State{$A_t\gets \lrc{v\in V_t\sm X_t\ |\ \mbox{all the children of } v \mbox{ in } T[V_t] \mbox{ are in } X_t}$}
    		\State{$U_t \gets \emptyset$} \Comment{$U_t$ stores the vertices covered in iteration $t$}
    		\For{$a \in A_t$}
        		\State{$T_a \gets T[\anc{T[V_t]}{a}\cup\des{T[V_t]}{a}]$} \label{line:Ta}
        		\State{$Q_1,\dots,Q_m\gets \textsc{NextFit}(a,T_a,k)$}
        		\State{Add $Q_1,\dots,Q_m$ to $\calC$} \Comment{Add the partial cover computed by $\textsc{NextFit}$}
        		\State{$U_t \gets U_t \cup Q_1\cup\cdots\cup Q_m$}
            \EndFor
    		\If {$V_t\sm U_t\ne \emptyset$}
    		      \State{$V_{t+1}\gets\anc{T[V_t]}{V_t\sm U_t}$}
    		\Else
    		      \State{$V_{t+1}\gets\emptyset$}
    		\EndIf
    		\State $t \gets t+1$
        \EndWhile
		\If {$V_t\ne \emptyset$} \Comment{The last set in the cover}
		      \State{Add $V_t$ to $\calC$}
		\EndIf
		\Return $\calC$
	\end{algorithmic}
\end{algorithm}
\negA
\begin{algorithm}[!ht]
	\caption{Next-Fit packing} 	\label{alg:nextfit}
	\begin{algorithmic}[1]
		\Procedure{NextFit}{$a,T_a,k$}
		\Statex \emph{Input}: An anchor $a\in A$, the subtree $T_a$ and an integer $k>0$
		\Statex \emph{Output}: A feasible cover $Q_1,\ldots,Q_m$ of some (potentially all) vertices in $\des{T_a}{a}$
		
		\State Let $u_1,\ldots,u_d$ be the children of $a$ in $T_a$
		\State $m \gets 1$
		\State $Q_m \gets \anc{T_a}{a}$
		
		\For{$s=1$ to $d$}
		\If{$s(Q_m)+s(\des{T_a}{u_s})\le k$}
		\State $Q_m \gets Q_m \cup \des{T_a}{u_s}$
		\Else
		\State $m \gets m+1$
		\State $Q_m \gets \anc{T_a}{a}\cup \des{T_a}{u_s}$
		\EndIf
		\EndFor
		\If{$m\mbox{ is odd}$} \Comment{Remove the subset $Q_m$ if $m$ is odd} \label{line:odd}
		\State $m\gets m-1$ \Comment{Note that $m>1$}
		\EndIf
		\Return $Q_1,\ldots,Q_m$
		\EndProcedure
	\end{algorithmic}
\end{algorithm}
Let $A=\bigcup_t A_t = \{a_1,a_2,\ldots\}$ be the set of anchors computed in all the iterations. For an anchor $a\in A$, let $t(a)$ be the iteration in which $a$ was added to the set of anchors. 
Note that any leaf $\ell$ of $T$ appears in exactly one subset in $\calC$. Thus, the iteration in which $\ell$ is {\em covered} is uniquely defined.

\begin{figure}
\centering
\begin{subfigure}[t]{0.49\textwidth}
\centering
\includegraphics[height=4.5cm, width=5cm]{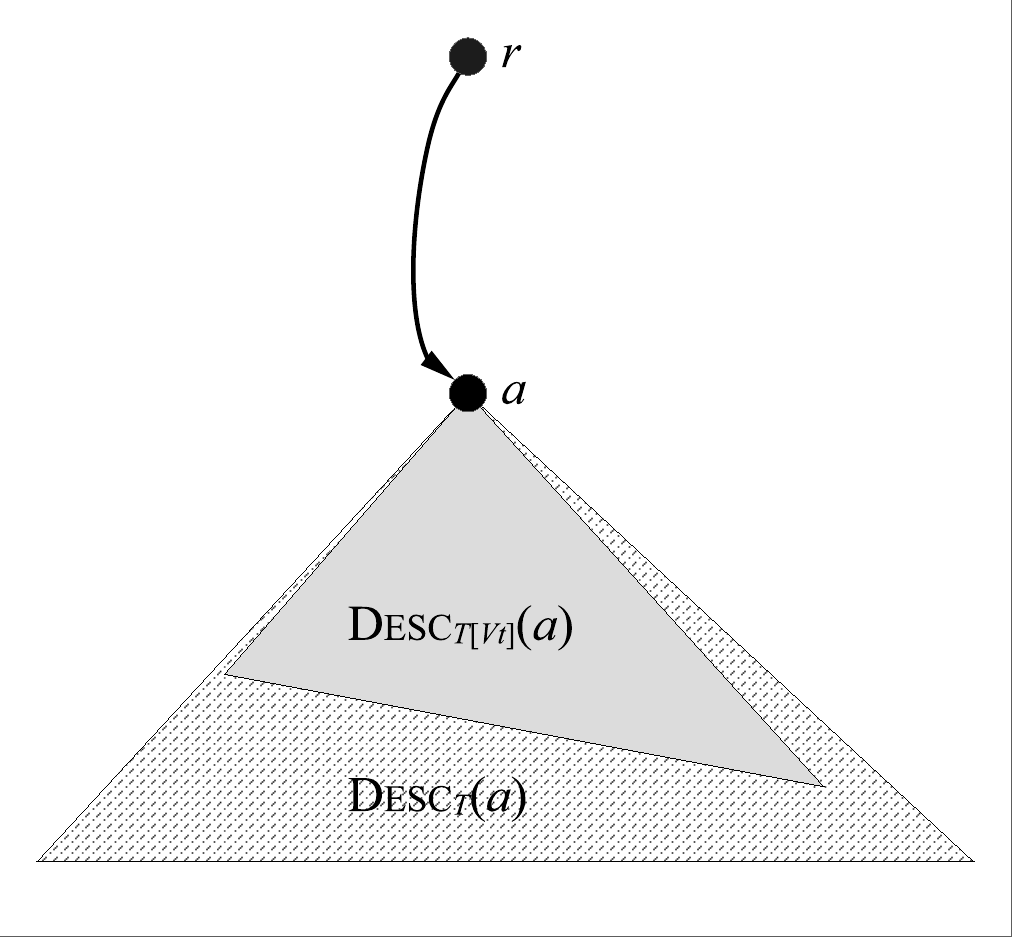}
\caption{\small The subtree $T_a$} 
\label{fig:Ta}
\end{subfigure}
\hfill
\begin{subfigure}[t]{0.49\textwidth}
\centering
\includegraphics[height=4.5cm, width=5cm]{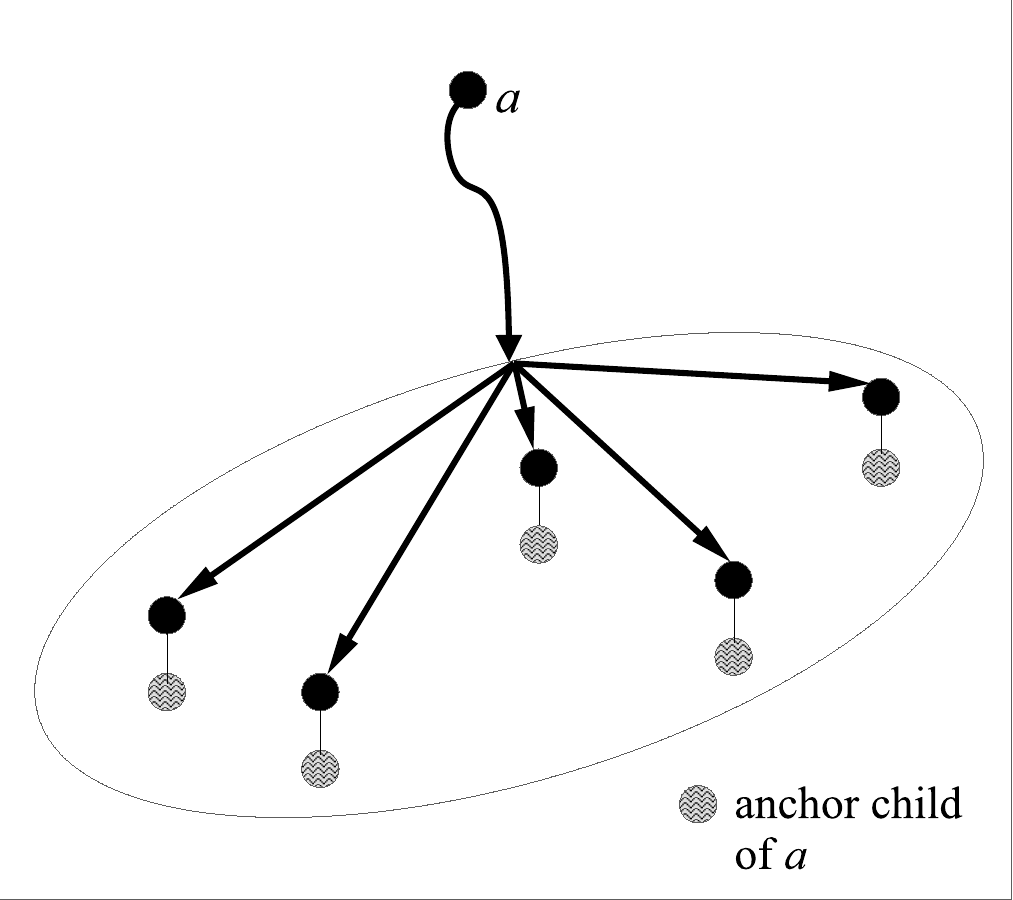} 
\caption{\small The subtree $S_a$} 
\label{fig:Sa}
\end{subfigure}
\hfill
\vspace{-10pt}
\caption{\small The subtrees $T_a$ defined in \cref{alg:anchor}, and $S_a$ defined in the proof of \cref{lem:LB}.}
\end{figure}

\begin{definition}
	Let $a\in A$ be an anchor.
	\begin{itemize}
		\item
		If $v\in \des{T_a}{a}$ is an ancestor of a leaf $\ell$ of $T$ that is covered in iteration $t(a)$ then we say that $v$ is {\em anchored} at $a$.
		\item
		If $v\in \des{T_a}{a}$ is not anchored at $a$ then we say that $v$ is a {\em leftover vertex} of $a$.
		\item
		Let $\na(a)$ denote the total size of the vertices that are anchored at $a$,
		and $\lo(a)$ denote the total size of the leftover vertices of $a$. 
	\end{itemize}
\end{definition}
  Clearly, $\na(a)+\lo(a)=s(\des{T_a}{a})$. Our assumption that for every leaf $\ell$ of $T$, $s(\ell)>0$, implies that (i) $\na(a)>0$ and (ii) if there are leftover vertices then $\lo(a)>0$.

The proofs of the next lemmas are in \cref{app:proofs}.
\begin{lemma}\label{lem:leftover}
    \comment{
	If $v$ is a  leftover vertex of $a$ then all the vertices on the path from $a$ to $v$ (excluding $a$) are also leftover vertices of $a$.
	If $v$ is anchored at $a$ then all the vertices on the path from $a$ to $v$ (excluding $a$) are also anchored at $a$.}
    Let $v\in \des{T_a}{a}$, and let $u_s$ be the (unique) child of $a$ that is also an ancestor of $v$.
	If $v$ is a  leftover vertex of $a$ then all the vertices in the subtree of $T_a$ rooted at $v$, as well as the vertices along the path from $u_s$ to $v$, are also leftover vertices of $a$.
	If $v$ is anchored at $a$ then all the vertices in the subtree of $T_a$ rooted at $v$, as well as the vertices along the path from $u_s$ to $v$, are also anchored at $a$. 
\end{lemma}
\comment{
\begin{proof}
	In \cref{alg:nextfit} we iterate over the children of $a$ in $T_a$ and for each such a child $u_s$ we add $u_s$ and all its descendants in $T_a$ to one of the subsets $Q_m$. Thus, if $Q_m$ is one of the subsets returned by Procedure NextFit (\cref{alg:nextfit}) when it is called for $a$, then $u_s$ and all its descendants in $T_a$ are anchored at $a$, 
    otherwise $u_s$ and all its descendants are leftover vertices of $a$. \qed
\end{proof}
}

\begin{lemma}\label{lem:at_most_one_anchor}
	For any two anchors $a,a'\in A$, the sets of vertices anchored at $a$ and $a'$ are disjoint.
\end{lemma}
\comment{
\begin{proof}
	To prove the lemma it suffices to show that a vertex cannot be anchored at more than one anchor.
	Fix an iteration $t$. Consider two anchors $a$ and $a'$ added in iteration $t$, namely, $t(a)=t(a')=t$. Note that $a$ is neither a descendant nor an ancestor of $a'$, and thus the set of vertices anchored at $a$, which is contained in $\des{T[V_{t}]}{a}$ is disjoint from set of vertices anchored at $a'$, which is contained in $\des{T[V_{t}]}{a'}$.
	To complete the proof, we note that for any two anchors $a$ and $a'$ such that $t(a)<t(a')$, the set of vertices that are anchored at $a$ are disjoint from the set of vertices that are anchored at $a'$. This follows since by \cref{lem:leftover} all the leaves of $T$ that are descendants of the vertices anchored at $a$ are covered at some iteration $t\le t(a)$ and thus the vertices anchored at $a$ are not in $V_{t(a)+1}$. \qed
\end{proof}
}

Define a ``parent-child'' relation among anchors as follows. For two anchors $a$ and $b$, we say that $a$ is the \textit{anchor-parent} of $b$ and $b$ is the \textit{anchor-child} of $a$ if (i) $a$ is an ancestor of $b$ in $T$, and (ii) the path from $a$ to $b$ (in $T$) does not contain any anchors other than $a$ and $b$. Note that if anchor $a$ is an anchor-parent of $b$ then $t(a)>t(b)$; that is, the iteration $t(a)$ in which $a$ is added to the set of anchors is after iteration $t(b)$. This follows from the definition of $A_t$.
For anchor $a\in A$, let $\AC(a)\ss \des{T}{a}\cap A$ be the set of anchor-children of $a$. We extend this definition for all $v\in V$, and let $\AC(v)\ss \des{T}{v}\cap A$ be all the anchors $b\in \des{T}{v}\cap A$ such that the path from $v$ to $b$ (in $T$) does not contain any anchors other than $b$ and (possibly) $v$. For $v\in V$, let
$\AD(v) = \des{T}{v}\cap A$ be the set of anchors that are also descendants of $v$.
A \textit{top} anchor is an anchor that is not an anchor-child of any other anchor. Let $\topA\sse A$ denote the set of top anchors. Note that if the root $r$ is an anchor then $\topA=\lrc{r}$.

\begin{lemma} \label{lem:alg_ub}
	The number of subsets in the solution computed by \Cref{alg:anchor} is upper bounded by
	\[
	\alpha+\sum_{a\in A} 2\floor*{\frac{\na(a)}{k-h(a)+1}},
	\]
	where
	\[
	\alpha =
	\begin{cases}
		1 & \exists a\in\topA \text{ s.t. } \lo(a)>0 \\
		1 & \exists \text{ leaf } \ell \in V \text{ s.t. } \anc{T}{\ell}\cap \topA = \emptyset\\
		0 & \text{otherwise}
	\end{cases}
	\]
\end{lemma}

\begin{proof}
	Let $\calQ = Q_1,\ldots,Q_d$ be the solution computed by the algorithm.
	Fix $a\in A$, and let $\calQ_a$ be the subsets in $\calQ$ that were returned by Procedure NextFit (\cref{alg:nextfit}) when it computed a feasible cover of the vertices in $\des{T_a}{a}$. Note that the union of all the subsets in $\calQ_a$ is the set of  vertices anchored at $a$ (whose total size is $\na(a)$) together with all the ancestors of $a$. Also, $\calQ_a$ consists of at least two subsets, and a vertex anchored at $a$ cannot appear in more than one subset in $\calQ_a$. Consider the subsets in $\calQ_a$ in the order in which they were computed by \cref{alg:nextfit}. It follows from Procedure NextFit that the total size of the vertices in any pair of consecutive subsets in this ordered list is at least $k-h(a)+1$. Since the total size of vertices anchored at $a$ is $\na(a)$, the number of such disjoint pairs is upper bounded by $\floor*{\frac{\na(a)}{k-h(a)+1}}$. By Line~\ref{line:odd} of Procedure NextFit, the number of sets in $Q_a$ is even, and thus the total number of subsets in $\calQ_a$ is upper bounded by $2\floor*{\frac{\na(a)}{k-h(a)+1}}$. The total upper bound is given by summing this bound over all anchors $a\in A$. We may have one additional subset if the algorithm is terminated when $|V_t|>0$. By our construction and \cref{lem:leftover}, this may happen {\em iff} there exists a leaf $\ell$ that is not anchored at any anchor. If such a leaf $\ell$ exists then one of the following two conditions must be satisfied: (i) $\ell$ has no ancestor that is an anchor, or  (ii) $\ell$ is a leftover vertex of the (unique) top anchor $a$ that is an ancestor of $\ell$, in which case $\lo(a)>0$. The lemma follows. \qed
\end{proof}

We now prove a lower bound on the number of subsets in any feasible solution and in particular in the optimal solution.
Let $\calP=P_1,\dots,P_p$ be a feasible solution. Since every subset in $\calP$ is closed under ancestor relation, some vertices may appear in multiple subsets. We refer to each such appearance of a vertex $v$ as an {\em occurrence} of $v$, and associate the size $s(v)$ to each of its occurrences. 
For an anchor $a\in A$, let $\calP(a)$ be the set of all the subsets in $\calP$ that contain vertices anchored at either $a$ or a descendant of $a$. For an anchor $a\in A$, let $\calU(a)\sse \des{T}{a}$ be the set of all vertices such that each vertex is both a descendant of $a$ and an ancestor of a vertex anchored at either $a$ or a descendant of $a$.

\begin{lemma} \label{lem:LB}
	For every $a\in A$,
	the number of subsets in $\calP(a)$ is at least
	\[
	\LB(a) = \sum_{b\in \AD(a)\cup \lrc{a}}
	\floor*{\frac{\na(b)}{k-h(b)}}.
	\]
	If the lower bound is tight then all the leaves that are in the subsets in $\calP(a)$  must be anchored at either $a$ or a descendant of $a$.
\end{lemma}
\begin{proof}
    To prove the lower bound of $\LB(a)$ on the {\em number} of subsets in $\calP(a)$ for every $a \in A$, we prove a slightly stronger lower bound of
    $\left(k-h(a)\right)\LB(a)$ on the {\em total size} of the occurrences of vertices in $\calU(a)$ in subsets in $\calP(a)$. Since any subset that contains a descendant of $a$ must contain also the ancestors of $a$ (including $a$), whose total size is $h(a)$, the total size of the vertices in $\calU(a)$ that can be in a single subset in $\calP(a)$ is no more than $k-h(a)$. Thus, a lower bound of $\left(k-h(a)\right)\LB(a)$ on the total size of the occurrences of vertices in $\calU(a)\sse \des{T}{a}$ in subsets in $\calP(a)$ implies a lower bound of $\LB(a)$ on the number of subsets in $\calP(a)$ (and on the number of occurrences of anchor $a$).
    
    The lower bound on the total size of the occurrences of vertices in $\calU(a)$ in subsets in $\calP(a)$ also implies that if the lower bound is tight then all the leaves that are in the subsets in $\calP(a)$ must be anchored at either $a$ or a descendant of $a$.  To see this, note that if any subset in $\calP(a)$  contains a leaf $\ell$ that is not anchored at an anchor in $\AD(a)\cup \lrc{a}$ then $\ell \notin \calU(a)$, also $s(\ell)>0$ by our assumption. It follows that the total size of the subsets in $\calP(a)$ is strictly more than $\left(k-h(a)\right)\LB(a)$. Clearly, this implies that the number of subsets in $\calP(a)$ is strictly more than $\LB(a)$.
    
	The proof is by induction starting from the \textit{bottom} anchors in $T$, which are the anchors with no anchor-children. 	
	For the induction base, consider a bottom anchor $a$. Note that in this case $\calU(a)$ is the set of all vertices anchored at $a$. The subsets in $\calP(a)$ cover all the vertices anchored at $a$; thus, the total size of the occurrences of these vertices in the subsets in $\calP(a)$ is at least $\na(a)$.
	Clearly, $\na(a) \ge \left(k-h(a)\right)\floor*{\frac{\na(a)}{k-h(a)}}=\left(k-h(a)\right)\LB(a)$.  
	
	For the inductive step, consider an anchor $a$ and assume that the lemma holds for every anchor $b\in\AC(a)$. Specifically, for every anchor $b\in\AC(a)$, the total size of the occurrences of vertices in $\calU(b)$ in subsets in $\calP(b)$ is at least $\left(k-h(b)\right)\LB(b)$. Note that $\calP(b) \sse \calP(a)$ and $\calU(b)\sse\calU(a)$.
	Let $S_a$ the subtree of $T$ rooted at $a$ given by the union of the paths from $a$ to each of its anchor-children, \textit{excluding} the anchor-children (see Figure~\ref{fig:Sa}). Note that the vertices of $S_a$ as well as the vertices in $\AC(v)$ are in $\calU(a)$.
    \begin{myclaim} \label{claim:treeSa}
		For every vertex $v$ of $S_a$, the total size of the occurrences of vertices in $\des{T}{v}\cap\calU(a)$ in subsets in $\bigcup_{b\in\AC(v)}\calP(b)$ is at least
		$\left(k-h(v)\right)\sum_{b\in \AC(v)} {\LB(b)}$.
	\end{myclaim}
	\begin{proof}[of Claim~\ref{claim:treeSa}]
		We prove the claim vertex by vertex, scanning the vertices of $S_a$ bottom-up. Consider a leaf $v$ of $S_a$. By the definition of $S_a$, its children are anchors in $\AC(v)$. By the induction hypothesis of \cref{lem:LB}, for every anchor $b\in\AC(v)$ the total size of the occurrences of vertices in $\calU(b)$ in the subsets in $\calP(b)$ is at least $\left(k-h(b)\right)\LB(b)$. 
        The total size of such occurrences that are contained in any single subset of $\calP(b)$ is at most $k-h(b)$, since any such subset must also contain the ancestors of $b$ (including $b$) whose size is $h(b)$. It follows that the {\em number} of occurrences of $b$ in these subsets in $\calP(b)$ is at least ${\LB(b)}$, and the total size of these occurrences is at least $s(b)\cdot {\LB(b)}$. Note that for any pair of anchors $b,b' \in \AC(v)$, the sets $\calU(b)$ and $\calU(b')$ are disjoint. Summing over all the anchor-children of $v$, we have that the total size of the occurrences of vertices in $\des{T}{v}\cap\calU(a)$ in the subsets in $\bigcup_{b\in\AC(v)}\calP(b)$ is at least
            $\sum_{b\in \AC(v)} \LB(b) \left( \left(k-h(b)\right) + {s(b)} \right) =   
		      \left(k-h(v)\right)\sum_{b\in \AC(v)} {\LB(b)}$.
        The last equality holds since for every $b\in \AC(v)$, $h(v)+s(b) = h(b)$.
		
		The lower bound for an internal vertex of $S_a$ is obtained similarly. Note that a child $u$ of $v$ is either an anchor or a vertex of $S_a$. If $u$ is an anchor, that is $u\in \AC(v)$, then as shown above, the total size of the occurrences of $u$ and its descendants in $\calU(u)$ in the subsets in $\calP(u)$ is
		$\left( k-h(u)+ s(u)\right){\LB(u)} =\left(k-h(v)\right) {\LB(u)}$. Suppose that $u$ is a vertex of $S_a$. Since $u$ is a child of $v$ and the vertices are scanned bottom-up the lower bound holds for $u$, and the total size of the occurrences of vertices in $\des{T}{u}\cap\calU(a)$ in the subsets in $\bigcup_{b\in\AC(u)}\calP(b)$ is
		$\left(k-h(u)\right)\sum_{b\in \AC(u)} {\LB(b)}$. The total size of such occurrences that is contained in any single subset in $\bigcup_{b\in\AC(u)}\calP(b)$ is at most $k-h(u)$, thus; the {\em number} of occurrences of $u$ in these subsets is at least $\sum_{b\in \AC(u)} {\LB(b)}$, and the total size of these occurrences is at least $s(u)\cdot \sum_{b\in \AC(u)} {\LB(b)}$. We get that the total size of the occurrences of $u$ and the vertices in $\des{T}{u}\cap\calU(a)$ in the subsets in $\bigcup_{b\in\AC(u)}\calP(b)$ is
		$\left(k-h(u)+s(u)\right)\sum_{b\in \AC(u)} {\LB(b)} =\left(k-h(v)\right) \sum_{b\in \AC(u)} {\LB(b)}$. For any pair $u,u'$ of children of $v$,  the sets $\des{T}{u}$ and $\des{T}{u'}$ are disjoint. Summing over all the children of $v$, we get that the total size of the occurrences of vertices in $\des{T}{v}\cap\calU(a)$ in the subsets in $\bigcup_{b\in\AC(v)}\calP(b)$ is at least
		$\left(k-h(v)\right)\sum_{b\in \AC(v)} {\LB(b)}$.\qed
	\end{proof}
\comment{
\begin{figure}
\begin{center}
\includegraphics[width=5cm]{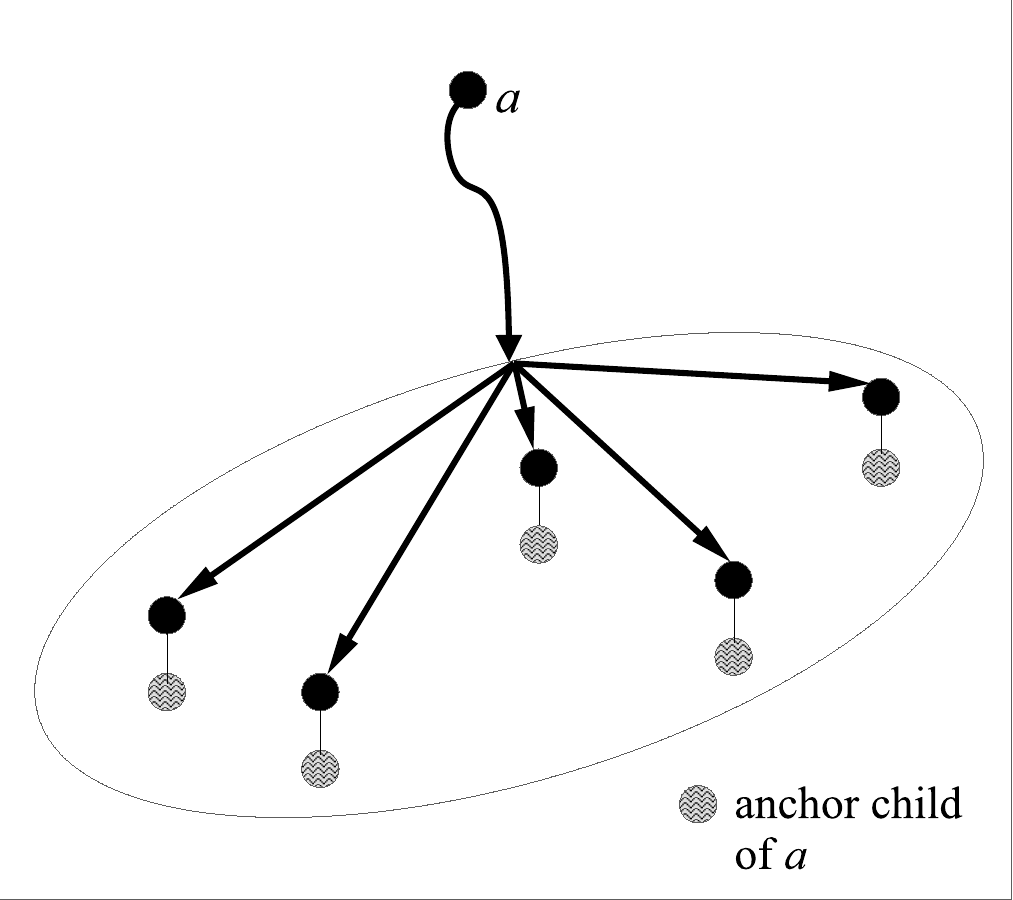} 
\caption{\small The subtree $S_a$ rooted at an anchor $a$, given as the set of paths from $a$ to its anchor-children.}
\label{fig:Sa}
\end{center}
\end{figure}
}
	
    Next, we consider vertices that are anchored at $a$. By the definition of $\calP(a)$, each such vertex $v$ must occur at least once in subsets in $\calP(a)$; also, $v\in\calU(a)$. Note that $v$ may be an ancestor of an anchor $b\in \AD(a)$. This may happen in case $v$ is a  vertex of $S_a$, and also in case a leftover vertex of an anchor $b\in\AD(a)$ is anchored at $a$, and $v$ is on the path from $a$ to $b$. In case $v$ is an ancestor of an anchor $b\in \AD(a)$, our induction hypothesis and Claim~\ref{claim:treeSa} already imply a lower bound on the number of its occurrences in subsets in $\calP(a)$. Specifically, in case $v\in\AD(a)$, our induction hypothesis implies a lower bound of ${\LB(v)}$ on the number of its occurrences, and in case $v\notin\AD(a)$ and $\AC(v)\ne\emptyset$, Claim~\ref{claim:treeSa} implies a lower bound of $\sum_{b\in \AC(v)} {\LB(b)}$ on the number of its occurrences. We prove that $v$ must occur at least {\em once more} in subsets in $\calP(a)$, {\em in addition} to this implied lower bound. This results in addition of $\na(a)$ to the total size of the occurrences of vertices anchored at $a$ in the subsets in $\calP(a)$.
	
	\begin{myclaim} \label{claim:anchored}
		For every vertex $v$ anchored at $a$, the number of occurrences of $v$ in the subsets in $\calP(a)$ is at least
		\[
		\begin{cases}
			1 & v \notin\anc{T}{\AD(a)} \\
			1+ {\LB(v)} & v\in\AD(a) \\
			1+\sum_{b\in \AC(v)} {\LB(b)} & \text{otherwise}
		\end{cases}
		\]
	\end{myclaim}
	\begin{proof}[of Claim~\ref{claim:anchored}]
		If $v$ is anchored at $a$ then it must be an ancestor of a leaf $\ell$ of $T$ that is anchored at $a$. Certainly, $v$ must occur in the subset in $\calP(a)$ that covers $\ell$. If $v$ is not an ancestor of an anchor $b\in \AD(a)$, we are done.
		
		If $v$ is an anchor and thus $v\in\AD(a)$, and the number of occurrences of $v$ in the subsets in $\calP(v)$ is strictly more than $\LB(v)$, then we are done. Otherwise,  the lower bound $\LB(v)$ is tight, and by the induction hypothesis of \cref{lem:LB}, all the leaves that are in the subsets in $\calP(v)$ must be anchored at $v$ or a descendant of $v$. Thus, none of these subsets can cover $\ell$. It follows that $v$ must occur in at least one more subset in $\calP(a)$ that covers $\ell$.
		
		A similar argument applies also if $v\notin\AD(a)$ and $\AC(v)\ne\emptyset$. Let $a'\in \AD(a)\cap\lrc{a}$ be the nearest ancestor of $v$ that is an anchor. By Claim~\ref{claim:treeSa}, the total size of the occurrences of vertices in $\des{T}{v}\cap\calU(a')$ in subsets in $\bigcup_{b\in\AC(v)}\calP(b)$ is at least
		$\left(k-h(v)\right)\sum_{b\in \AC(v)} {\LB(b)}$. It follows that the number of occurrences of $v$ in the subsets in $\bigcup_{b\in\AC(v)}\calP(b)$ is at least $\sum_{b\in \AC(v)} {\LB(b)}$. If $\ell$ is not in any of the subsets in $\bigcup_{b\in\AC(v)}\calP(b)$, then $v$ must occur in at least one more subset in $\calP(a)$ that covers $\ell$, and we are done. Suppose that this is not the case, and $\ell$ is in a subset in $\bigcup_{b\in\AC(v)}\calP(b)$. It is not difficult to verify that the proof of Claim~\ref{claim:treeSa} implies the lower bound on the total size of the vertices in a {\em subset} of $\des{T}{v}\cap\calU(a')$. This subset is the union of three sets: 
        $\des{T}{v}\cap\left(\bigcup_{b\in\AC(v)}\calU(b)\right)$, $\AC(v)$, and
        $\des{S_{a'}}{v}$. Clearly, $\ell$ is not in any of these three sets. Thus, the total size of the occurrences of vertices in $\des{T}{v}\cap\calU(a)$ in the subsets in $\bigcup_{b\in\AC(v)}\calP(b)\ss \calP(a)$ is strictly more than $\left(k-h(v)\right)\sum_{b\in \AC(v)} {\LB(b)}$. Hence, the number of occurrences of $v$ in the subsets in $\bigcup_{b\in\AC(v)}\calP(b)$ is strictly more than $\sum_{b\in \AC(v)} {\LB(b)}$.\qed
	\end{proof}
	
	By Claims~\ref{claim:treeSa} and~\ref{claim:anchored} and our induction hypothesis we get that the total size of the occurrences of vertices in $\des{T}{a}\cap\calU(a)$ in the subsets in $\calP(a)$ is at least
	\begin{align*}
		\na(a) &+\left(k-h(a)\right)\sum_{b\in \AC(a)} {\LB(b)}
		  \ge
		\left(k-h(a)\right)\left(\floor*{\frac{\na(a)}{k-h(a)}}+\sum_{b\in \AC(a)} {\LB(b)}\right)\\
		& = \left(k-h(a)\right)\sum_{c\in \AD(a)\cup \lrc{a}} \floor*{\frac{\na(c)}{k-h(c)}} = \left(k-h(a)\right)\LB(a)
	\end{align*}
\comment{
    To complete the proof of \cref{lem:LB} we show that if a leaf $\ell$ of $T$ that is not anchored at anchor in $\AD(a)\cup \lrc{a}$ is in a subset in $\calP(a)$ then the total size of the subsets in $\calP(a)$ is strictly more than $\left(k-h(a)\right)\LB(a)$. We distinguish two cases. If $\ell$ is in $\bigcup_{b\in\AC(a)}\calP(b)$, then by our induction hypothesis there exists an anchor $b\in \AC(a)$ such that the total size of the subsets in $\calP(b)$ is strictly more than $\left(k-h(b)\right)\LB(b)$. By the analysis above the total size of the subsets in $\calP(a)$ is at least $\na(a)$ plus the total size of the occurrences of all vertices in $\bigcup_{b\in\AC(a)}\calP(b)$. We get that the total size of the subsets in $\calP(a)$ is strictly more than $\left(k-h(a)\right)\LB(a)$. If $\ell$ is in $\calP(a) \sm \bigcup_{b\in\AC(a)}\calP(b)$, then the analysis above implies that the total size of the subsets in $\calP(a)$ is at least $s(\ell)+\na(a)+\left(k-h(a)\right)\sum_{b\in \AC(a)} {\LB(b)}$, and thus it is strictly more than $\left(k-h(a)\right)\LB(a)$.
}
The first equality holds since ${\LB(b)}=\sum_{c\in \AD(b)\cup \lrc{b}} \floor*{\frac{\na(c)}{k-h(c)}}$. \qed
\end{proof}

\begin{corollary} \label{cor:optlb}
	The number of subsets in any feasible solution is at least
	\[
	\alpha+\sum_{a\in A} \floor*{\frac{\na(a)}{k-h(a)}},
	\]
	where $\alpha$ is defined in \cref{lem:alg_ub}.
 \comment{
	\[
	\alpha =
	\begin{cases}
		1 & \exists a\in\topA \text{ s.t. } \lo(a)>0 \\
		1 & \exists \text{ leaf } \ell \in V \text{ s.t. } \anc{T}{\ell}\cap \topA = \emptyset\\
		0 & \text{otherwise}
	\end{cases}
	\]
 }
\end{corollary}
\begin{proof}
    If $\topA =\lrc{r}$ then by \cref{lem:LB} the number of occurrences of $r$ is at least $\LB(r)=\sum_{a\in A} \floor*{\frac{\na(a)}{k-h(a)}}$. If this lower bound is tight then all the leaves in subsets in $\calP(r)$ are anchored at some vertex. If $\lo(r)>0$, then there is a leaf of $T$ that is a leftover vertex of $r$ and thus not anchored at any vertex. In this case, at least one additional subset is needed to cover this leaf. If $\topA \ne\lrc{r}$ then $\AC(r)=\topA$. In this case, following the proof of Claim~\ref{claim:treeSa}, we get that the total size of the occurrences of the descendants of $r$ in subsets in $\calP$ is at least $\left(k-s(r)\right)\sum_{a\in \topA} \LB(a)$. This implies that $r$ occurs in at least $\sum_{a\in \topA} \LB(a) = \sum_{a\in A} \floor*{\frac{\na(a)}{k-h(a)}}$ subsets in $\calP$. If the bound is tight then all the leaves in these subsets are anchored vertices. Thus, if there exists a vertex that is not anchored at any vertex, at least one additional subset is needed. This occurs when either $\exists a\in\topA \text{ s.t. } \lo(a)>0$, or $\exists \text{ leaf } \ell \in V$ s.t. none of the ancestors of $\ell$ is an anchor. \qed
\end{proof}

Corollary~\ref{cor:optlb} and \cref{lem:alg_ub} imply a factor 2 approximation.

%% file: open_probs.tex
\mysection{Open problems}
\label{sec:open_probs}
An intriguing open problem is to bridge the gap between our $2$-approximation and $1.5$-inapproximability result for {\sc ct}. 
Recall that {\sc ct} is the special case of {\CPO} on out-trees.  
While we expect {\CPO} to be hard to approximate on general graphs (as
mentioned above), exploring further the hardness of {\CPO} on various graph classes 
remains open.

Another appealing line of research is to investigate the connections between {\sc cpo} and a natural covering variant of the {\DKSH} problem defined as follows.
Given a hypergraph $G=(V, E)$ and an integer $k$, find the minimum number of vertex sets,
each of cardinality at most $k$, such that every hyperedge is fully contained in one of the sets. We are not aware of earlir studies of this problem, even in the special case where $G$ is a graph. 
One interesting direction is to derive nontrivial hardness results for this problem and 
show possible implications for {\CPO}.

%% file: motivation.tex
\section{Motivation for {\sc rcp}} 
\label{sec:motivation}



A prime motivation for studying {\sc rcp} comes from the area of networking \cite{dong2015rule,yan2014cab,sheu2016wildcard,huang2015cost,li2019tale,stonebraker1990rules,li2015fdrc,gao2021ovs,cheng2018switch,rottenstreich2016optimal,li2020taming,gamage2012high,yan2018adaptive,rottenstreich2020cooperative,rastegar2020rule,yang2020pipecache}. In a {\em Software-Defined Network} (SDN) traffic flow is governed by a logically centralized controller that utilizes packet-processing {\em rules} to manage the underlying switches \cite{katta2016cacheflow}. The number of rules tends to be high while most traffic relies on a small fraction of these rules~\cite{sarrar2012leveraging}. Thus, caching frequently used rules can accelerate the processing time of the packets. 
However, standard caching policies cannot be used due to {\em dependencies} among rules.
One common form of dependency is a partial overlap in the binary strings representing the rules. For example, consider the rules $R_1$=`10**' (where the symbol `*' denotes a wildcard) and $R_2$=`1000'. Then whenever $R_1$ is placed in the cache, $R_2$ must be placed as well. Indeed, if only $R_1$ is in the cache then a message with a header `1000' would be matched with $R_1$, causing a correctness issue in handling this packet. 
Now, the problem of placing a feasible subset of the rules which handle a maximum total volume of traffic can modeled as follows.
We represent the rules by a DAG $G=(V,E)$, where $v_i \in V$ corresponds to  the rule $R_i$, and there is a directed edge from $v_i$ to $v_j$ if placing $R_j$ in the cache implies that $R_i$ is also in the cache. 
The profit of each vertex $v_i \in V$  reflects the volume of traffic handled by the rule $R_i$. The goal is to select a subset of vertices of maximum total profit which fits into the cache, that is closed under precedence constraints.


{\RCP} can be used also to model the {\em maximal extractable value} (MEV) problem in {\em blockchain}~\cite{mcmenamin2022fairtradex,obadia2021unity,weintraub2022flash,ARL19}. 
Each blockchain transaction is associated with a fee earned by the miner who creates the block containing this transaction. The set of transactions is associated with a partial order, and each blockchain prefix has to be closed under precedence constraints.
MEV is the maximum potential profit that a blockchain miner can gain from transactions that have not been validated.
Computing MEV can be cast as an {\RCP} instance where the vertices of the graph are the transactions, the edges represent the precedence constraints, the profits are the associated fees, and the bound $k$ is the number of transactions that fit in a single block. Other applications of {\RCP} variants arise, e.g., in the mining industry~\cite{moreno2010large,samavati2017methodology} and in scheduling~\cite{papazachos2010gang,woeginger2001approximability,efsandiari2015approximate,ibarra1978approximation,lenstra1977complexity}.

%% file: hardness.tex
\section{Hardness Result for CT}
\label{sec:hard}




Our hardness result for {\CT} is based on a reduction from {\em bin packing with cluster complement conflict graph} ({\BPCC}). An undirected graph $G = (V,E)$ is called a {\em cluster complement} if there is a partition $V_1,\ldots, V_m$ of $V$ such that for all $i \in [m]$ it holds that $V_i$ is an independent set in $G$ and for all $i,j \in [m]$ where $i \neq j$ and any $v \in V_i$ and $u \in V_j$ it holds that $\{u,v\} \in E$. We now formally define the {\BPCC} problem. 

\begin{definition}
	\label{def:BPCC}
	The {\em bin packing with cluster complement conflict graph ({\BPCC})} is defined as follows.
	\newline	\noindent {\bf Input:} A cluster complement $G = (V,E)$, a weight function $w:V \rightarrow \mathbb{Z}^{+}_{0}$, and a value $k \in \mathbb{N}$. 
	\newline	\noindent {\bf Configuration:} An independent set $C \subseteq V$ in $G$ such that $w(C) \leq k$.
	\newline \noindent {\bf Solution:} For some $m \in \mathbb{N}$, we say that $\left(C_1,\ldots, C_q\right)$ is a {\em solution with cardinality $q$} if the following holds.
	\begin{itemize}
		\item For every $i \in [q]$ it holds that $C_i$ is a configuration.
		\item For all $v \in V$ there is $i \in [q]$ such that $v \in C_i$.
	\end{itemize}
	\noindent {\bf Objective:} Find a solution of minimum cardinality.   
\end{definition}

\noindent\textbf{Proof of \Cref{thm:hard}:} 
We show a reduction from {\BPCC} to {\CT}. Let $I = (G = (V,E),w,k)$ be a {\BPCC} instance. Let $V_1,\ldots, V_m$ be the unique partition of $V$ into maximal independent sets, which exists and can be found in polynomial time since $G$ is cluster complement. Then, define the {\em reduced} {\CT} instance $X_I = (H = (\cV,\cE),s,K)$ as follows
\begin{itemize}
	\item The vertex set $\cV$ of $X_I$ contains a root $r$ and a vertex $r_i$ for every $i \in [m]$, where $(r,r_i) \in \cE$ ($r_i$ is a child of $r$ for every $i \in [m]$). For every $i \in [m]$ and every $v \in V_i$ define a leaf $\ell_{v}$ and add an edge $(r_i,\ell_v) \in \cE$. Overall, we get a two-level star graph.   
	\item Define the size function $s:\cV \rightarrow \mathbb{Z}^{+}_0$ such that $s(r) = 0$, for all $i \in [m]$ define $s(r_i) = 2 \cdot k$, and for all $i \in [m]$ and $v \in V_i$ define $s(\ell_v) = w(v)$. 
	\item Define $K = 3 \cdot k$. 
\end{itemize}

For every $C \subseteq V$, let 
\begin{equation}
	\label{eq:XC}
	X(C) = \{r\} \cup \bigcup_{i \in [m]~|~C \cap V_i \neq \emptyset} \{r_i\} \cup \bigcup_{v \in C} \{\ell_v\}. 
\end{equation}

\begin{myclaim}
	\label{claim:hard}
	For every $C \subseteq V$if $C$ is a configuration of $I$ then $X(C)$ is a configuration of $X_I$. 
\end{myclaim}
\begin{claimproof}
	Assume that $C$ is a configuration of $I$. Observe that, by \eqref{eq:XC}, $X(C)$ is closed under the precedence constraints. Moreover, 
	$$s(X(C)) = s(r)+\sum_{i \in [m]~|~C \cap V_i \neq \emptyset} s(r_i)+\sum_{v \in C} s(\ell_v) = 0+2k+w(C) \leq 3 \cdot k = K.$$
	The first equality follows from~\eqref{eq:XC}. The second equality holds since $C$ is a configuration; thus, it is an independent set in $G$, and it can contain vertices from at most one $V_i, i \in [m]$. The inequality holds since $C$ is a configuration. We conclude that $X(C)$ is a configuration of $X_I$. 
\end{claimproof}

For every $C \subseteq \cV$ let 
\begin{equation}
	\label{eq:XC'}
	I(C) = \bigcup_{\ell_v \in C~|~v \in V} \{v\}.  
\end{equation}

\begin{myclaim}
	\label{claim:hard2}
	For every $C \subseteq \cV$ if $C$ is a configuration of $X_I$ then $I(C)$ is a configuration of $I$. 
\end{myclaim}
\begin{claimproof}
	Assume that $C$ is a configuration of $X_I$. Then, for all $i,j \in [m]$, $i \neq j$, and $v \in V_i, u \in V_j$ it holds that $\ell_{v} \notin X(C)$ or $\ell_u \notin X(C)$; otherwise, by the precedence constraints we must include both $r_i$ and $r_j$ in $C$, and we have 
	$$s(C) \geq s(r_i)+s(r_j) = 4 \cdot k>3k = K,$$
	i.e., $C$ is not a configuration. Contradiction.
	Let $i \in [m]$ be the unique index such that $V_i \cap C \neq \emptyset$ (if there is no such index the proof is trivial). Then, 
	$$w(I(C)) = \sum_{v \in I(C)} s(\ell_v) =  \sum_{v \in I(C)} s(\ell_v)+s(r_i)-s(r_i) \leq s(C)-s(r_i) \leq K-2 \cdot k = k.$$
	The last inequality holds since $C$ is a configuration of $X_I$. We conclude that $I(C)$ is a configuration of $I$ by \Cref{def:BPCC}. 
\end{claimproof}

Let $S = (S_1,\ldots, S_q)$ be the solution for the {\CT} instance $X_I$. By Claim~\ref{claim:hard2}, for all $i \in [q]$ it holds that $I(S_i)$ is a configuration of $I$; thus, $I(S) = (I(S_1),\ldots, I(S_q))$ is a solution for $I$, since it contains every $v \in V$ at least once by \eqref{eq:XC'} and the definition of $X_I$.  Conversely, let $S = (S_1,\ldots, S_q)$ be a solution for the {\BPCC} instance $I$. Then, by Claim~\ref{claim:hard} for all $i \in [q]$ it holds that $X(S_i)$ is a configuration of $X_I$. Thus, $X(S) = (X(S_1),\ldots, X(S_q))$ is a solution for $I$, since it contains every $v \in \cV$ at least once by \eqref{eq:XC}.  Thus, there is a an asymptotic $\alpha$-approximation for {\BPCC} if and only if there is a an asymptotic $\alpha$-approximation for {\CT}. Hence, the claim of the theorem follows from
Lemma G.1 in~\cite{doron2023approximating}. 
\qed

%% file: reduction.tex
\section{A Reduction from D$k$SH to Rule Caching}
\label{sec:2}

In this section we give the first direction for the proof of Theorem~\ref{thm:EQ} and the proof of Corollary~\ref{cor:RCP}. Let $I = (G,p,k)$ denote an {\sc rcp} instance where 
$G = (V,E)$ is a directed graph, $p$ is the profit function $p:V \rightarrow \mathbb{Z}_0^+$ and $k \in \mathbb{N}$ specifies the maximal number of vertices allowed in a feasible solution. 

Our result is based on a reduction from the {\em densest $k$-subhypergraph} ({\DKSH}) problem  to {\RCP}. We start with some definitions and notations. 
A hypergraph is a pair $H = (V_H,E_H)$ where  $V_H$ is a set of vertices and $E_H \subseteq 2^{V_H}$ is a set of hyperedges, which are subsets of vertices.  For some $S \subseteq V_H$, let $H[S] = (S,E^S_H)$ be the induced subhypergraph of $S$ where the hyperedges of the subgraph are $E^S_H = \{e \in E_H~|~ e \subseteq S\}$; that is, all hyperedges contained in $S$. In the {\DKSH} problem the input is a tuple $(H,k,w)$ where $H = (V_H, E_H)$ is a hypergraph, $k$ is an integer parameter, and $w: E_H \rightarrow \mathbb{Z}_{\geq 0}$ is a weight function on the hyperedges. A feasible solution is  
a set of at most $k$ vertices. The objective is to find a feasible solution $S \subseteq V_H$ with maximum total weight of hyperedges in the subgraph induced by
this set, that is $w\left(E^S_H\right) = \sum_{e \in E^S_H} w(e)$.  

Informally, the main idea of our reduction is to represent a {\DKSH} instance $\mathcal{H}$ as an {\RCP} instance $R(\mathcal{H})$, where each vertex and hyperedge of $\mathcal{H}$ are vertices in the new instance. Furthermore, we duplicate the vertices of $\mathcal{H}$ a sufficiently large number of times (see Figure~\ref{fig:M1}). This ensures that the effect of taking the hyperedges to the solution of $R(\mathcal{H})$ as vertices of high profit, uses only a negligible portion of the cardinality bound corresponding to $R(\mathcal{H})$. This is formalized as Lemma~\ref{1}, leading to the proof of Theorem~\ref{thm:EQ}. In particular, our hardness result holds already for the special case of {\RCP} where the graph $G$ is directed-bipartite, and there are only two distinct profits for the vertices. 

For the remainder of this section, let $\mathcal{H} = (H,k,w)$ be a {\DKSH} instance where $H = (V_H,E_H)$. For the reduction from {\DKSH} to {\RCP}, we now define an {\RCP} instance based on $\mathcal{H}$. Let $m = |E_H|$ and let $U = \{v^i~|~v \in V_H, i \in [m+1]\}$ be a set of vertices, where each vertex in $V_H$ has $m+1$ distinct {\em copies}; these copies are  considered as distinct vertices in the reduced graph. Let $V = U \cup E_H$ be a set of vertices containing all copies of the vertices in $V_H$ and all {\em edge-vertices}, which are hyperedges in $E_H$. The {\em reduced graph} of $H$ is a directed graph $G(H) = (V, E)$, where

\begin{equation}
	\label{Em}
	E = \{(v^i,e)~|~e \in E_H, v \in V_H \cap e, i \in [m+1]\}.
\end{equation} 
That is, each copy of a vertex $v \in V_H$ has an outgoing edge to any edge-vertex that contains $v$ in $H$. See Figure~\ref{fig:M1} for an example of the construction of the reduced graph given a hypergraph $H$. By \eqref{Em}, we infer that the only edges in $G(H)$ are between	a vertex from $U$ to an edge-vertex from $E_H$. Thus, $G(H)$ is a DAG and a directed-bipartite with the bipartition $(U, E_H)$.

\begin{figure}
	\centering
	
	\begin{tikzpicture}
		\tikzstyle{edge} = [->, line width=1pt]
		\node (v1) at (0,2) {};
		\node (v5) at (9,5) {};
		\node (v6) at (9,4.5) {};
		\node (v7) at (9,4) {};
		
		\node (v8) at (9,3.5) {};
		\node (v9) at (9,3) {};
		\node (v10) at (9,2.5) {};
		
		\node (v11) at (9,2) {};
		\node (v12) at (9,1.5) {};
		\node (v13) at (9,1) {};
		
		\node (v14) at (9,0.5) {};
		\node (v15) at (9,0) {};
		\node (v16) at (9,-0.5) {};
		
		\node (v17) at (13,3.5) {};
		\node (v18) at (13,1.5) {};
		
		\node (v2) at (1.5,3) {};
		\node (v3) at (4,2.5) {};
		\node (v4) at (0,0.5) {};
		\draw[edge] (v5) to (v17);
		\draw[edge] (v6) to (v17);
		\draw[edge] (v7) to (v17);
		\draw[edge] (v8) to (v17);
		\draw[edge] (v9) to (v17);
		\draw[edge] (v10) to (v17);
		\draw[edge] (v11) to (v17);
		\draw[edge] (v12) to (v17);
		\draw[edge] (v13) to (v17);
		
		\draw[edge] (v11) to (v18);
		\draw[edge] (v12) to (v18);
		\draw[edge] (v13) to (v18);
		\draw[edge] (v14) to (v18);
		\draw[edge] (v15) to (v18);
		\draw[edge] (v16) to (v18);
		
		\begin{scope}[fill opacity=0.9]
			\filldraw[fill=yellow!70] ($(v1)+(-0.5,0)$) 
			to[out=90,in=180] ($(v2) + (0,0.5)$) 
			to[out=0,in=90] ($(v3) + (1,0)$)
			to[out=270,in=0] ($(v2) + (1,-0.8)$)
			to[out=180,in=270] ($(v1)+(-0.5,0)$);
			
			\filldraw[fill=green!50] ($(v4)+(-0.5,0.2)$)
			to[out=90,in=180] ($(v4)+(0,2)$)
			to[out=0,in=90] ($(v4)+(0.6,0.3)$)
			to[out=270,in=0] ($(v4)+(0,-0.6)$)
			to[out=180,in=270] ($(v4)+(-0.5,0.2)$);
			
			\comment{arg1
				\filldraw[fill=blue!70] ($(v4)+(-0.5,0.2)$)
				to[out=90,in=180] ($(v4)+(0,1)$)
				to[out=0,in=90] ($(v4)+(0.6,0.3)$)
				to[out=270,in=0] ($(v4)+(0,-0.6)$)
				to[out=180,in=270] ($(v4)+(-0.5,0.2)$);
				\filldraw[fill=green!80] ($(v5)+(-0.5,0)$)
				to[out=90,in=225] ($(v3)+(-0.5,-1)$)
				to[out=45,in=270] ($(v3)+(-0.7,0)$)
				to[out=90,in=180] ($(v3)+(0,0.5)$)
				to[out=0,in=90] ($(v3)+(0.7,0)$)
				to[out=270,in=90] ($(v3)+(-0.3,-1.8)$)
				to[out=270,in=90] ($(v6)+(0.5,-0.3)$)
				to[out=270,in=270] ($(v5)+(-0.5,0)$);
				\filldraw[fill=red!70] ($(v2)+(-0.5,-0.2)$) 
				to[out=90,in=180] ($(v2) + (0.2,0.4)$) 
				to[out=0,in=180] ($(v3) + (0,0.3)$)
				to[out=0,in=90] ($(v3) + (0.3,-0.1)$)
				to[out=270,in=0] ($(v3) + (0,-0.3)$)
				to[out=180,in=0] ($(v3) + (-1.3,0)$)
				to[out=180,in=270] ($(v2)+(-0.5,-0.2)$);
			}
		\end{scope}

		\foreach \v in {1,2,...,18} {
			\fill (v\v) circle (0.1);
		}
		
		\filldraw[fill=brown!70] (v1) circle (0.1) node [right] {$v_3$};
		\filldraw[fill=darkgray!70] (v2) circle (0.1) node [below left] {$v_2$};
		\filldraw[fill=blue!70] (v3) circle (0.1) node [left] {$v_1$};
		
		\filldraw[fill=blue!70](v5) circle (0.1) node [left] {$v^1_1$};
		\filldraw[fill=blue!70] (v6) circle (0.1) node [left] {$v^2_1$};
		\filldraw[fill=blue!70] (v7) circle (0.1) node [left] {$v^3_1$};
		
		\filldraw[fill=darkgray!70] (v8) circle (0.1) node [left] {$v^1_2$};
		\filldraw[fill=darkgray!70] (v9) circle (0.1) node [left] {$v^2_2$};
		\filldraw[fill=darkgray!70] (v10) circle (0.1) node [left] {$v^3_2$};
		
		\filldraw[fill=brown!70] (v11) circle (0.1) node [left] {$v^1_3$};
		\filldraw[fill=brown!70] (v12) circle (0.1) node [left] {$v^2_3$};
		\filldraw[fill=brown!70] (v13) circle (0.1) node [left] {$v^3_3$};
		
		\filldraw[fill=red!70] (v14) circle (0.1) node [left] {$v^1_4$};
		\filldraw[fill=red!70] (v15) circle (0.1) node [left] {$v^2_4$};
		\filldraw[fill=red!70] (v16) circle (0.1) node [left] {$v^3_4$};
		
		\filldraw[fill=yellow!70] (v17) circle (0.1) node [right] {$e_1$};
		\filldraw[fill=green!70] (v18) circle (0.1) node [right] {$e_2$};
		
		\filldraw[fill=red!70] (v4) circle (0.1) node [below] {$v_4$};
		\comment{arg1
			\fill (v5) circle (0.1) node [below right] {$v_5$};
			\fill (v6) circle (0.1) node [below left] {$v_6$};
			\fill (v7) circle (0.1) node [below right] {$v_7$};
		}
		\node at (0.2,2.8) {$e_1$};
		
		\node at (0,1.2) {$e_2$};
		\comment{arg1
			\node at (2.3,3) {$e_2$};
			\node at (3,0.8) {$e_3$};
			
		}
	\end{tikzpicture}
	\caption{An illustration of the reduction from {\DKSH} to {\RCP}. On the left, there is a hypergraph $H = (V_H, E_H)$ where $V_H = \{v_1,v_2,v_3,v_4\}$ and $E_H = \{e_1 = \{v_1,v_2,v_3\}, e_2 = \{v_3,v_4\}\}$. The reduced graph $G(H)$ on the right contains $|E_H|+1 = 3$ copies of each vertex in $H$. } 
	\label{fig:M1}
\end{figure}
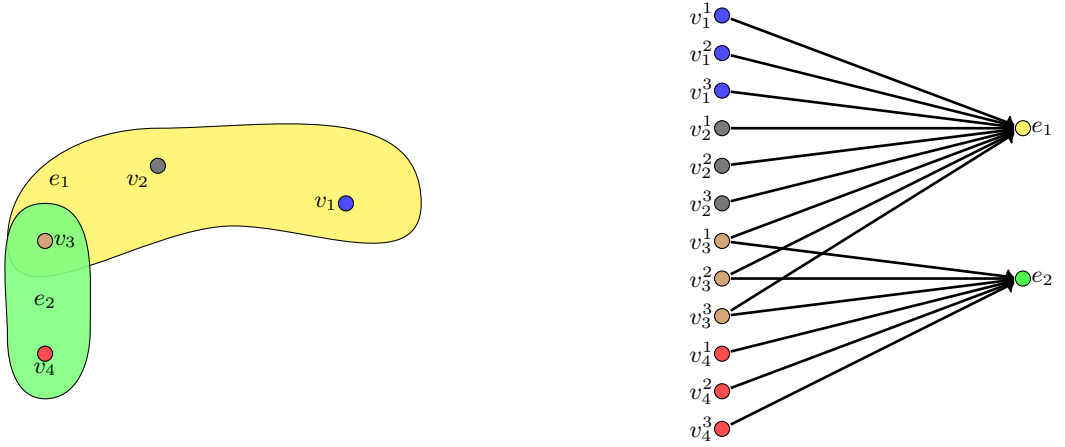

The profit of each edge-vertex in $G$ is equal to the weight of the corresponding hyperedge in $H$, and the profit of all copies of the vertices of $H$ is equal to zero. 
Formally, let $p:V \rightarrow \zzp$ be the profit function. Then, $p(u) = 0$ fo all
$u \in U$ $p(u) = 0$, and for all $e \in E_H$, $p(e)$ is equal to the weight of the hyperedge $e \in E_H$. Finally, let $c = k \cdot (m+1)+m$ be the cardinality bound for the {\RCP} instance of the reduction. Now, let $R(\mathcal{H}) = (G,p,c)$ be the {\em reduction} of $\mathcal{H}$. 

For a subset of vertices $S \subseteq V_H$, let $E_H^S=\{ e \in E_H|~ e \subseteq S \}$ be the set of edges in $H$ contained in $S$.

\begin{lemma}
	\label{1}
	If there is a feasible solution $S$ for the {\DKSH} instance $\mathcal{H}$ such that $w \left( E^S_H \right) = q$, then the \textnormal{\RCP} instance $R(\mathcal{H})$ has a solution of profit at least $q$.  
\end{lemma}

\begin{proof}
	Let $S \subseteq V_H$ such that $|S| \leq k$ and the induced subhypergraph $H[S]$ has total weight $q$; that is, $w \left( E^S_H \right) = q$. Define
	
	\begin{equation}
		\label{T}
		T = E^S_H \cup \{u^i \in U ~|~ u \in S, i \in [m+1] \}. 
	\end{equation} In words, $T$ contains all edge-vertices which are subsets of $S$ (i.e., the set $E^S_H$), and all copies of vertices in $S$.

	\begin{myclaim}
		\label{clm:T}
		$T$ is a feasible solution for $R(\mathcal{H})$.
	\end{myclaim}
	
	\begin{proof}
		We note that  
		\begin{equation}
			\label{eqT}
			\begin{aligned}
				|T| \leq{} & |E^S_H|+|\{u^i \in U ~|~ u \in S, i \in [m+1] \}| \\
				={} & q+ |S|  \cdot (m+1) \leq q+ k \cdot (m+1) \\
				\leq{} & m+ k \cdot (m+1)  \\
				={} & c.
			\end{aligned}
		\end{equation}
		The first inequality is by \eqref{T}. The second inequality holds since $S$ is a feasible solution for $\mathcal{H}$; thus, $|S| \leq k$. The third inequality holds since $E^S_H$ is a subset of $E_H$; hence, $\left|E^S_H\right| \leq |E_H| = m$. By \eqref{eqT}, $T$ contains at most $c$ vertices. To show that $T$ is feasible, we show that there are no inedges to $T$ from $V \setminus T$. First, there are no inedges to any $u \in U$ by \eqref{Em}. Second, for all $e \in T \cap E_H$ and $(u,e) \in E$, by \eqref{Em} and \eqref{T} it holds that $u \in T$; that is, we take all copies of vertices with inedges to $e$ by the definition of $T$. We conclude that $T$ is a feasible solution for $R(\mathcal{H})$.
	\qed\end{proof}
	
	We now show that the total profit of $T$ is at least $q$. $$\sum_{t \in T} w(t) \geq \sum_{e \in E_H \cap T} w(e) = w \left( E_H \cap T \right) = w \left(E^S_H\right) = q.$$ The second equality is by \eqref{T}. \qed
\end{proof} 

For some $T \subseteq V$, we define the subset of vertices in $H$ for which all copies in $G(H)$ are in $T$. That is, 

\begin{equation}
	\label{S}
	S(T) = \{v \in V_H  ~|~ \forall i \in [m+1]: v^i \in T \}. 
\end{equation}

The next claims are used in the proof of the second direction in the reduction (stated as Lemma~\ref{2}). 

\begin{myclaim}
	\label{claim:1}
	For all $T \subseteq V$ it holds that $|S(T)| \cdot  (m+1) \leq |U \cap T|$. 
\end{myclaim}

\begin{proof}
	By \eqref{S}, it holds that $T \cap U$ contains exactly $m+1$ copies of each $v \in S(T)$. Thus, as vertex copies form distinct vertices in $G(H)$, the claim follows.  
	\comment{arg1
		We define an injective function $f:S(T) \rightarrow U_m \cap T$ in the following way. For all $v \in S(T)$ define $f(v) = v^1$. Since $v \in S(T)$ then by \eqref{S} it follows that for all $i \in [m+1]$ we have $v^i \in T$. In particular, $v^1 \in T \cap U_m$ and it follows that $f$ is well defined. 
		
		To show that $f$ is injective, let $u,v \in S(T)$ such that $u \neq v$. Therefore, $u^1 \neq v^1$ since copies of distinct vertices are also distinct. Hence, we have $f(u) = u^1 \neq v^1 = f(v)$ and $f$ is injective. 
	}
\qed\end{proof}

\begin{myclaim}
	\label{claim:2}
	For all feasible solutions $T \subseteq V$ of $R(\cH)$ it holds that $|S(T)| \leq k$. 
\end{myclaim} 

\begin{proof}
	Assume towards a contradiction that  $|S(T)| > k$. Therefore, 
	\begin{equation}
		\label{eq1}
		\begin{aligned}
			|T| ={} & |U \cap T|+|E_H \cap T| \\
			\geq{} & |S(T)|\cdot (m+1)+|E_H \cap T| \\
			\geq{} &  |S(T)|\cdot (m+1) \\ 
			\geq{} & (k+1) \cdot (m+1) \\
			={} & k \cdot (m+1)+(m+1) \\
			>{} & k \cdot (m+1)+m \\
			={} & c. 
		\end{aligned}
	\end{equation}
	The first equality follows from $E_H \cap U = \emptyset$.  The first inequality is by Claim~\ref{claim:1}. The third inequality holds since $|S(T)| > k$; thus, since $S(T)$ is a set, the cardinality of $S(T)$ is an integer and it follows that  $|S(T)| \geq k+1$. By \eqref{eq1} we reach a contradiction to the feasibility of $T$
	for $R(\mathcal{H})$. 
\qed\end{proof}

Our proof of Theorem~\ref{thm:EQ} relies on the next result. 
\begin{lemma}
	\label{2}
	If the \textnormal{\RCP} instance $R(\mathcal{H})$ has a feasible  solution $T$ of profit $q$, then $S(T)$ is a feasible solution for the {\DKSH} instance $\mathcal{H}$ such that $w \left(E^{S(T)}_H\right) \geq q$.
\end{lemma}

\begin{proof}
	Let $T \subseteq V$ be a feasible  solution of profit $q$ for $R(\mathcal{H})$; that is, $|T| \leq c$ and $p(T) = \sum_{t \in T} p(t) = q$. We first show that $S(T)$ is a feasible solution for $\mathcal{H}$.
	We note that $S(T) \subseteq V_H$, by \eqref{S}. Also, it holds that $|S(T)| \leq k$ by Claim~\ref{claim:2}. Now, we show that $w \left(E^{S(T)}_H\right) \geq q$ using the following claim.
	\begin{myclaim}
		\label{claim:3}
		$w \left(\{e \in E_H~|~ e \subseteq S(T)\} \right) \geq p \left(E_H \cap T \right)$. 
	\end{myclaim}
	
	\begin{proof}
		Let $e \in E_H \cap T$ and let $v \in e$. By \eqref{Em}, we have that $(v^i,e) \in E$ for all $i \in [m+1]$. Since $e \in T$ and $T$ is a feasible solution for the {\RCP} instance $R(\mathcal{H})$, for all $i \in [m+1]$ it holds that $v^i \in T$ (else $T$ is not a feasible solution since $e \in T, v^i \notin T$ and $(v^i,e) \in E$).  Therefore, by the above and using \eqref{S}, we conclude that $v \in S(T)$.
		This holds for all $v \in e$; therefore, $e \in S(T)$. Hence, $  E_H \cap T \subseteq \left\{e \in E_H~|~ e \subseteq S(T)\right\}$ and the claim follows. 
	\qed\end{proof}
	
	Therefore,
	$$w\left(E^{S(T)}_H\right) = w \left(\{e \in E_H~|~ e \subseteq S(T)\}\right) \geq p \left(E_H \cap T \right) = q.$$ The first inequality holds by Claim~\ref{claim:3}. The second equality holds since (i) for all $u \in U$ and  $e \in E_H$ it holds that $p(u) = 0$, and (ii) the profit of the solution $T$ is $q$; therefore, $p \left(E_H \cap T \right) = q$. \qed
\end{proof}

We are ready to prove the main result of this section.  
\\
\noindent {\bf Proof of Theorem~\ref{thm:EQ} ($\Longrightarrow$):} 
Assume that for some $\rho \geq 1$, there is a $\rho$-approximation algorithm $\cA$ for {\RCP}. Given a {\DKSH} instance $\mathcal{H} = (H,k,w)$, we construct the instance $R(\cH)$ for {\RCP}, and run on $R(\cH)$ Algorithm $\cA$. Let $S^*$ be an optimal solution for {\DKSH} on $\cH$, and let $q^* = w \left(E^{S^*}_H\right)$. Now, by Lemma~\ref{1}, there is a solution for $R(\mathcal{H})$ of profit at least $q^*$. Therefore, Algorithm $\cA$ returns a feasible solution $T_{q^*}$ for $R(\mathcal{H})$ of profit at least $q^* / \rho$. Hence, by Lemma~\ref{2}, we have that $S(T_{q^*})$ is a feasible solution for $\mathcal{H}$ such that $w \left( E^{S(T_{q^*})}_H\right) \geq q^* / \rho$. Note that $\cA$ is polynomial and that the construction of $S(T_{q^*})$ can be trivially computed in polynomial time. Hence, we have a $\rho$-approximation for {\DKSH}. \qed

The proof of Corollary~\ref{cor:RCP} follows from Theorem~\ref{thm:EQ} and the results of~\cite{manurangsi2017inapproximability,hajiaghayi2006minimum}.

\noindent {\bf Proof of Corollary~\ref{cor:RCP}:}
By Theorem~\ref{thm:EQ}, an $n^{1-\eps}$-approximation for {\RCP} implies an $n^{1-\eps}$-approximation for {\DKSH}, where $n$ is the number of vertices in the {\RCP} instance. Thus, by the results of  \cite{hajiaghayi2006minimum}, this implies an $n^{1-\eps}$-approximation for the {\em Maximum Balanced Biclique problem in bipartite graphs (MBB)}. Hence, assuming SSEH and that $NP \neq BPP$, we reach a contradiction to a result of \cite{manurangsi2017inapproximability}. \qed

\section{Reducing Rule Caching to D$k$SH}
\label{sec:6}

In this section, we give the second direction in the proof of Theorem~\ref{thm:EQ}; namely, we show that given an approximation algorithm for {\DKSH}, we can solve the {\sc rule caching problem} with the same approximation guarantee. 

In our reduction, we restrict ourselves to {\em minimal} {\DKSH} solutions, in which every vertex appears in some hyperedge in the solution. Formally, 
\begin{definition}
	\label{def:minimal}
	Let $\cH = (H,k,w), H = (V_H,E_H)$ be a \textnormal{DkSH} instance. A solution $S \subseteq V_H$ is called a {\em minimal solution} of $\cH$ if for every $v \in S$ it holds that $\{e \in E_H~|~ e \subseteq S, v \in e\} \neq \emptyset$. 
\end{definition} 

We show that it suffices to consider minimal solutions without losing in the number of hyperedges in the solution. Recall that for a hypergraph $H = (V_H,E_H)$ and $S \subseteq V_H$ we use the notation $E_H^S = \{e \in E_H~|~ e \subseteq S\}$. 

\begin{lemma}
	\label{lem:minimal}
	for any $\rho \geq 1$, if there is a $\rho$-approximation for \textnormal{DkSH} then there is a $\rho$-approximation for \textnormal{DkSH} that always returns a minimal solution. 
\end{lemma}

\begin{proof}
	Let $\rho \geq 1$ and let $\cA$ be a $\rho$-approximation algorithm for {\DKSH}. We define the following algorithm $\cB$ based on $\cA$. Let $\cH = (H,k,w), H = (V_H,E_H)$ be a {\DKSH} instance.
	\begin{enumerate}
		\item Compute $S \leftarrow \cA(\cH)$.
		\item Return $T = \big\{v \in S~|~\{e \in E~|~ e \subseteq S, v \in e\} \neq \emptyset \big\}$. 
	\end{enumerate}
	Clearly, the running time is polynomial since $\cA$ is polynomial. Moreover, $T$ is a minimal solution for $\cH$. By the definition of $T$, it holds that $w \left(E_H^S \right) = w\left(E_H^T \right)$. This completes the proof.  \qed
\end{proof}

In the following, we use an approximation algorithm for {\DKSH} that produces minimal solutions to achieve an approximation algorithm for {\RCP}. The reduction considers the set of {\em predecessor} vertices of each vertex $v$ in the {\RCP} instance; these are all vertices from which there are directed paths to $v$ in the graph. Formally,

\begin{definition}
	\label{def:pre}
	Let $G = (V,E)$ be a directed graph. Let $\mathcal{Z}_G$ be all pairs of vertices $(u,v) \in V \times V$ such that there is a directed path from $u$ to $v$ in $G$ (this includes $(v,v)$). For every $v \in V$, define the {\em predecessors} of $v$ as $P_G(v) = \{u \in V~|~(u,v) \in \mathcal{Z}_G\}$. 
\end{definition} When $G$ is clear from the context, we simply use $P(v) = P_G(v)$. In our reduction we take an {\RCP} instance with a graph $G = (V,E)$ and construct a {\DKSH} instance with a hypergraph $H$ such that the vertices of $H$ are the vertices of $G$ and the hyperedges of $H$ are the sets of predecessors $P(v)$ for every $v \in V$; the weight of $P(v)$ is the profit of $v$ in the {\RCP} instance. 

\begin{definition}
	\label{def:RCPreduction}
	Let 	$I = (G,p,k), G = (V,E)$,  be an \textnormal{\RCP} instance. Define the {\em reduced}  \textnormal{DkSH} instance $\cH_I = (H_I,w_I,k)$ of $I$ such that the following holds.
	\begin{itemize}
		\item $H_I = (V,\mathcal{P}_I)$, where $,\mathcal{P}_I = \left\{ P(v)~|~v \in V\right\}$.
		\item $w_I :\mathcal{P}_I \rightarrow \zzp$ such that for all $v \in V$ it holds that $w_I \left( P(v) \right) = p(v)$. 
	\end{itemize}
\end{definition}

\begin{figure}
	\centering
	
	\begin{tikzpicture}
		\tikzstyle{edge} = [->, line width=1pt]
		\node (v1) at (0.5,2.5) {};
		\node (v2) at (1.5,3) {};
		\node (v3) at (4,2.5) {};
		\node (v4) at (0.5,1) {};
		
		\node (1) at (8.5,2.5) {};
		\node (2) at (10,2.5) {};
		\node (3) at (8.5,1) {};
		\node (4) at (7,2.5) {};
		%
		\draw[edge] (4) to (1);
		\draw[edge] (1) to (2);
		\draw[edge] (3) to (2);
		\begin{scope}[fill opacity=0.7]
			\fill[fill=blue!70] ($(v1)+(-1,0)$) 
			to[out=90,in=180] ($(v2) + (0,0.5)$) 
			to[out=0,in=90] ($(v3) + (1,0)$)
			to[out=270,in=0] ($(v2) + (1,-1.8)$)
			to[out=180,in=270] ($(v4)+(-1,-1)$);
			
			\fill[fill=white!50] ($(v4)+(-0.5,0.2)$)
			to[out=90,in=180] ($(v4)+(0,2)$)
			to[out=0,in=90] ($(v4)+(0.6,0.3)$)
			to[out=270,in=0] ($(v4)+(0,-0.4)$)
			to[out=180,in=270] ($(v4)+(-0.5,1)$);
		\end{scope}

		
		
		\filldraw[fill=brown!70] (v1) circle (0.1) node [right] {$1$};
		\filldraw[fill=yellow!70] (v2) circle (0.1) node [right] {$2$};
		\filldraw[fill=white!70] (v3) circle (0.1) node [right] {$3$};
		\filldraw[fill=red!70] (v4) circle (0.1) node [right] {$4$};
		
		
		\filldraw[fill=brown!70](1) circle (0.1) node [above] {$1$};
		\filldraw[fill=yellow!70] (2) circle (0.1) node [above] {$2$};
		\filldraw[fill=white!70] (3) circle (0.1) node [below] {$3$};
		\filldraw[fill=red!70] (4) circle (0.1) node [above] {$4$};
		%
		%
		%
		%

		\node at (2.2,3.7) {$P(2)$};			
		\node at (0.5,1.7) {$P(1)$};
	\end{tikzpicture}
	\caption{An illustration of the reduction from {\RCP} to {\DKSH}. On the right side there is a graph $G = (V,E)$ of an {\RCP} instance $I$. The hypergraph $H_I = (V,\cP_I)$ is on the left (hyperedges consisting of one vertex are not explicitly shown).} \label{fig:M2}
\end{figure}
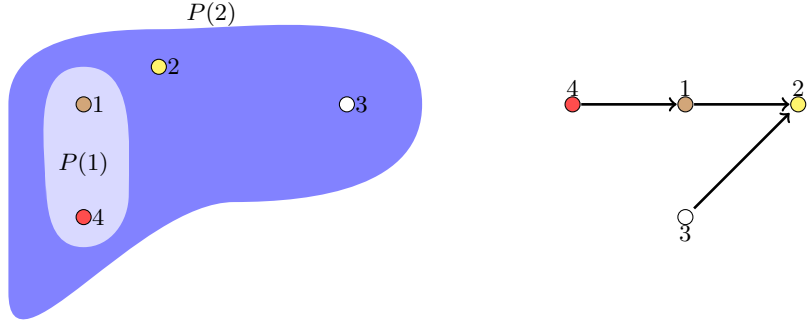

See Figure~\ref{fig:M2} for an illustration of the reduction. The next lemma summarizes the reduction. For the second direction, note that it is necessary that the {\DKSH} solution would be minimal to ensure feasibility. 

\begin{lemma}
	\label{lem:RCPreduction}
	Let $I = (G,p,k), G = (V,E)$,  be an \textnormal{\RCP} instance, and let $\cH_I = (H_I,w_I,k)$ be the reduced \textnormal{DkSH} instance of $I$. For every $W \in \zzp$, $S \subseteq V$ is a solution for $I$ of profit $W$ if and only if $S$ is a minimal solution for $\cH_I$ of weight $W$. 
\end{lemma}

\begin{proof}
	We prove the two directions of the lemma. 
	\begin{enumerate}
	\item[(i)]
	For the first direction, let $S \subseteq V$ be a solution for $I$ of profit $W$. By Definition~\ref{def:RCPreduction}, it holds that $S$ is a solution for $\cH_I$. We use $\cP^S_I$ to denote the set of hyperedges of $H_I$ contained in $S$. Then, by Definition~\ref{def:RCPreduction}, for all $v \in S$, it holds that $P(v) \in \cP^S_I$ and $v \in P(v)$ by Definition~\ref{def:pre}; Thus, by Definition~\ref{def:minimal}, $S$ is a minimal solution for $\cH_I$. Finally, the weight of $S$ in $\cH_I$ satisfies $$w_I\left(\cP^S_I\right) = \sum_{e \in  \cP^S_I} w_I(e) = \sum_{v \in S} w_I\left(P(v)\right) = \sum_{v \in S} p(v) = p(S) = W,$$
	as required.
\item[(ii)]	
	For the second direction of the proof, let $T$ be a minimal solution for $\cH_I$ of weight $W$. Observe that $|T| \leq k$; thus, to show that $T$ is a solution for the {\RCP} instance $I$, consider some $v \in T$. As $T$ is a minimal solution for $\cH_I$, there is $x \in T$ such that $P(x) \subseteq T$ and $v \in P(x)$. Because $P(v) \subseteq P(x)$, it follows that $P(v) \subseteq T$. Hence, $T$ is closed under precedence constraints and is a solution for $I$ as well. To conclude, we note that the profit of $T$ for $I$ satisfies 
	$$p(T) = \sum_{v \in T} p(v) = \sum_{v \in T} w_I\left(P(v)\right) =  \sum_{e \in \cP_I \text{ s.t. } e \subseteq T} w_I(e) = w_I\left(\cP^T_I\right)  = W.$$
	The third equality holds since we showed that, for all $v \in T$, $P(v) \subseteq T$. \qed
\end{enumerate} 
\end{proof}

By Lemmas~\ref{lem:minimal} and~\ref{lem:RCPreduction} we have the proof of Theorem~\ref{thm:EQ}. 

\noindent {\bf Proof of Theorem~\ref{thm:EQ} ($\Longleftarrow$):} 
For some $\rho \geq 1$, let $\cA$ be a $\rho$-approximation algorithm for {\DKSH}. By Lemma~\ref{lem:minimal}, there is a $\rho$-approximation algorithm $\cB$ for {\DKSH} that always returns a minimal solution. For an {\RCP} instance $I$, we can run $\cB$ on the reduced {\DKSH} instance $\cH_I$ and return the output. By Lemma~\ref{lem:RCPreduction}, this yields a $\rho$-approximation for {\RCP}.  \qed

%% file: uniform.tex
\section{Hardness of Rule Caching with Uniform Profits}
\label{sec:uniform}
In this section we give a hardness result for {\sc uniform rcp (u-rcp)}, the special case of {\RCP} where all vertices have uniform profits, i.e., the proof of Theorem~\ref{thm:2}. The proof is based on a reduction from the {\em densest} $k$-subgraph problem, that is the special case of {\DKSH} (see Section~\ref{sec:2}) where all hyperedges contain exactly two vertices. This problem is hard to solve even when {\em parametrized} by~$k$.

Towards presenting our result, we give some formal definitions for approximation schemes. 	Let $|I|$ be the encoding size of an instance~$I$ of a problem $\Pi$. 
	A {\em polynomial-time approximation scheme} (PTAS)
	for $\Pi$ is a family of algorithms $(A_{\eps})_{\eps>0}$ such that, for any $\eps>0$, $A_{\eps}$ is a polynomial-time $(1+ \eps)$-approximation algorithm for $\Pi$. 
	An {\em Efficient PTAS} (EPTAS) is a PTAS $(A_{\eps})_{\eps>0}$ with running time of the form $f\left(\frac{1}{\eps}\right) \cdot |I|^{O(1)}$, where $f$ is an arbitrary computable function. 
	 An approximation scheme 
	$(A_{\eps})_{\eps>0}$ is a {\em Fully PTAS} (FPTAS) if the running time of $A_{\eps}$ is of the form $ {\left(\frac{|I|}{\eps}\right)}^{O(1)}$.

Parameterized complexity analyzes algorithms for hard computational problems by considering their running time as function of both the input size and a specific parameter $k \in \N$. A problem is classified as {\em fixed-parameter tractable} (FPT) if it can be solved in time $f(k) \cdot \text{poly}(n)$, where $f$ is a function dependent solely on $k$ and $n$ is the input size, allowing for exponential dependence on the parameter $k$. For further definitions and concepts relating to parameterized complexity, see, e.g.~\cite{downey2013fundamentals,cygan2015parameterized}.  
The next lemma asserts that {\URCP} is still challenging, despite the uniform profits. 
\begin{lemma}
	\label{lem:uniform}
	If there is an \textnormal{EPTAS} for \textnormal{\URCP} then there is an optimal \textnormal{FPT} algorithm for densest $k$-subgraph parametrized by $k$. 
\end{lemma} 
We give the proof below. 

\noindent{\bf Proof of Theorem~\ref{thm:2}:}
 It is known that densest $k$-subgraph parametrized by $k$ does not admit an FPT algorithm~\cite{chalermsook2017gap}. Thus, the theorem follows from Lemma~\ref{lem:uniform}
\qed

For the proof of Lemma~\ref{lem:uniform}, we use the following construction of an input for {\URCP}, given an input graph $G$ for the densest $k$-subgraph problem. Informally, each vertex $v$ in the undireced graph $G$ is duplicated $2m$ times, for some $m \in \mathbb{N}$, where there are bidirectional edges between all copies of the vertex. This guarantees that either all copies of $v$ or none of them are taken for the solution. In addition, we take the edges of $G$ as vertices of the constructed directed graph. Each edge-vertex has an incoming edge from each of its vertex endpoints in $G$. See Figure~\ref{fig:M0} for an illustration of the construction. Formally, 
\begin{definition}
	\label{def:m-reduced}
	Given an undirected graph $G = (V,E)$ and a parameter $m \in \mathbb{N}$, we define $D_m(G) = (L_m \cup E, \bar{E}_m)$ as the $m$-reduced graph of $G$ such that the following holds. \begin{itemize}
		\item $L_m = \{v_i~|~v \in V, i \in [2\cdot m]\}$. 
		\item $\bar{E}^m_1 = \{\left(v_i, v_j\right)~|~i,j \in [2\cdot m], v \in V\}$. 
		\item $\bar{E}^m_2 = \{\left(v_i, e\right)~|~i \in [2\cdot m], v_i \in L_m, e \in E, v \in e\}$. 
		\item $\bar{E}_m = \bar{E}^m_1 \cup \bar{E}^m_2$. 
		
	\end{itemize}
\end{definition}

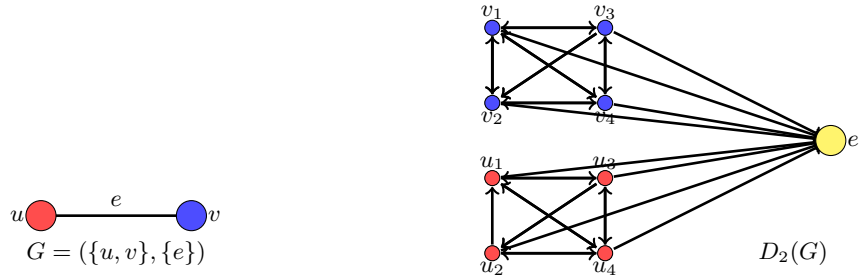
\begin{figure}
	\centering
	
	\begin{tikzpicture}
		\tikzstyle{edge} = [->, line width=1pt]
		\tikzstyle{uedge} = [-, line width=1pt]
		\node (v3) at (0,2) {};
		\node (e) at (10.5,3) {};
		\node (v5) at (6,4.5) {};
		\node (v6) at (6,3.5) {};
		\node (v7) at (7.5,4.5) {};
		
		\node (v8) at (7.5,3.5) {};

		\node (u5) at (6,2.5) {};
		\node (u6) at (6,1.5) {};
		\node (u7) at (7.5,2.5) {};
		
		\node (u8) at (7.5,1.5) {};
		
		\node (v4) at (2,2) {};
		
		\draw[edge] (v5) to (v6);
		\draw[edge] (v6) to (v5);
		\draw[edge] (v5) to (v7);
		\draw[edge] (v7) to (v5);
		\draw[edge] (v5) to (v8);
		\draw[edge] (v8) to (v5);
		\draw[edge] (v7) to (v6);
		\draw[edge] (v7) to (v6);
		\draw[edge] (v8) to (v6);
		\draw[edge] (v6) to (v8);
		\draw[edge] (v7) to (v8);
		\draw[edge] (v8) to (v7);

		\draw[edge] (v5) to (v6);
		\draw[edge] (u6) to (u5);
		\draw[edge] (u5) to (u7);
		\draw[edge] (u7) to (u5);
		\draw[edge] (u5) to (u8);
		\draw[edge] (u8) to (u5);
		\draw[edge] (u7) to (u6);
		\draw[edge] (u7) to (u6);
		\draw[edge] (u8) to (u6);
		\draw[edge] (u6) to (u8);
		\draw[edge] (u7) to (u8);
		\draw[edge] (u8) to (u7);
		
		\draw[edge] (u5) to (e);
		\draw[edge] (u6) to (e);
		\draw[edge] (u7) to (e);
		\draw[edge] (u8) to (e);
		\draw[edge] (v5) to (e);
		\draw[edge] (v6) to (e);
		\draw[edge] (v7) to (e);
		\draw[edge] (v8) to (e);
		
		\draw[uedge] (v3) to (v4);

		\filldraw[fill=yellow!70] (e) circle (0.2) node [right] {$~e$};
		
		\foreach \v in {3,4,...,8} {
			\fill (v\v) circle (0.1);
			
		}
		
		\foreach \u in {5,6,...,8} {
			\fill (u\u) circle (0.1);
		}
		
		\filldraw[fill=red!70] (v3) circle (0.2) node [left] {$ ~u~ $};
		\filldraw[fill=blue!70] (v4) circle (0.2) node [right] {$ ~v$};
		
		\filldraw[fill=blue!70](v5) circle (0.1) node [above] {$v_1$};
		\filldraw[fill=blue!70] (v6) circle (0.1) node[below] {$v_2$};
		\filldraw[fill=blue!70] (v7) circle (0.1) node[above] {$v_3$};
		\filldraw[fill=blue!70] (v8) circle (0.1) node [below] {$v_4$};
		
		\filldraw[fill=red!70](u5) circle (0.1) node  [above] {$u_1$};
		\filldraw[fill=red!70] (u6) circle (0.1) node [below] {$u_2$};
		\filldraw[fill=red!70] (u7) circle (0.1) node  [above] {$u_3$};
		\filldraw[fill=red!70] (u8) circle (0.1) node [below] {$u_4$};

		
		%
		\node at (1,2.2) {$e$};
		
		\node at (1,1.5) {$G = (\{u,v\}, \{e\})$};

		\node at (10,1.5) {$D_2(G)$};
		
		
	\end{tikzpicture}
	\caption{An illustration of the reduction from densest $k$-subgraph to {\URCP}. On the left there is a simple undirected graph $G$ with a single edge. The $2$-reduced directed graph of $G$ is on the right. Each vertex of $G$ is replaced by $2\cdot2 = 4$ copies  
	with a bidirectional edge connecting any two copies of the same vertex, and an outgoing edge from each copy to the single edge-vertex $e$.} \label{fig:M0}
\end{figure}

\noindent{\bf Proof of Lemma~\ref{lem:uniform}:}
Let $\cI = (G,k)$ be a densest $k$-subgraph instance, where $G = (V,E)$ is an undirected graph, and $k$ is the cardinality bound for the subset of vertices in the solution.
 Also, let $\cA$ be an EPTAS for {\URCP}. Observe that the value of the optimum of $\cI$ is an integer $0 \leq m^* \leq \binom{k}{2}$. We also assume that $k \geq 2$ and $m^* \geq 1$, else the problem is trivial. 
Hence, in time $O(k^2)$ we can iterate over all values $m \in \left\{1,\ldots, \binom{k}{2} \right\}$. Henceforth, we assume a fixed choice for $m \in \{1,\ldots, \binom{k}{2}\}$. Now, define the {\URCP} instance $\cj_m = \left(D_m(G),h_m\right)$, where $h_m = 2k \cdot m+m$ is the cardinality bound in the {\URCP} instance.

Next, we run the EPTAS $\cA$ on $\cj_m$ with error parameter $\eps_m = \frac{1}{2h_m}$. Let $\cA(\cj_m,\eps_m) =U(m)$ be the output for $\cj_m$ and $\eps_m$ (i.e., a subset of vertices in $D_m(G)$). 
\begin{myclaim}
	\label{clm:m*}
	$ |U(m^*) \cap E| \geq m^*$.
\end{myclaim}

\begin{proof}
	Assume towards contradiction that $ |U(m^*) \cap E| < m^*$. 
 W.l.o.g. we assume that $k \leq |V|$. Then,

\begin{equation}
\label{eq:Ubound}
    \begin{aligned}
        \left|U(m^*)\right| ={} & \left|U(m^*) \cap L_{m^*}\right|+\left|U(m^*) \cap E\right| \\
        \leq{} & 2m^* \cdot k+ \left|U(m^*) \cap E\right| \\
        <{} & 2m^* \cdot k+ m^*. 
    \end{aligned}
\end{equation}

 The first inequality holds by Definition~\ref{def:m-reduced}: for each $v \in V$ it holds that either all copies of $v$ or none of them belong to $U(m^*)$. In addition, there can be at most $k$ different vertices $v \in V$ for which all copies are contained in $U(m^*)$, otherwise $|U(m^*)| \geq (k+1) 2 m^*> h_{m^*}$.
	
	Since $m^*$ is the value of the optimum for $\cI$, there is a solution  $F \subseteq V$ for $\cI$, $|F| \leq k$, such that $\left|\{e \in E~|~ e \subseteq F\}\right| = m^*$. Moreover, by the definition of $k$, there is $P \subseteq V \setminus F$ such that $|P| = k-|F|$. Now, define $H = \{v_i~|~i \in [2\cdot m^*], v \in F \cup P\}$ and $Q = H \cup \{e \in E~|~ e \subseteq F\}$.  Observe that $Q \subseteq L_{m^*} \cup E$. Furthermore, it holds that $Q$ is closed under the precedence constraints of $D_{m^*}(G)$ by the definition of $Q$ and Definition~\ref{def:m-reduced}. 
	Finally, it holds that 
	\begin{equation}
	\label{eq:Lbound}
	|Q| = |Q \cap L_{m^*}|+|Q\cap E| = |H|+|Q \cap E| = 2m^* \cdot k+m^* = h_{m^*}.
	\end{equation}
	 By \eqref{eq:Ubound} and \eqref{eq:Lbound}, $Q$ is a feasible solution for $\cj_{m^*}$
	 with a value strictly larger than the value of $U(m^*)$. Since $\eps_{m^*}$ is sufficiently small, we have that
	
	
	\begin{equation*}
		\begin{aligned}
			\left| U(m^*) \right| \geq{} & (1-\eps_{m*}) \OPT(\cj_{m^*}) \\
			\geq{} & (1-\frac{1}{2 h_{m^*}}) |Q| \\
			={} & \left(1-\frac{1}{2 \cdot (2k \cdot m^*+m^*)}\right) \cdot \left(2m^* \cdot k+m^* \right) \\
			={} & 2m^* \cdot k+m^* - \frac{1}{2} 
		\end{aligned}
	\end{equation*} 
	
As $|U(m^*)|$ is integral, we have $|U(m^*)| = 2m^* \cdot k+m^*$. Contradiction to 
\eqref{eq:Ubound}. Now, we choose $m'$ which satisfies 
\begin{equation}
	\label{eq:m'}
	m'=\argmax_{m'' \in  \left\{1,\ldots, \binom{k}{2} \right\}} \left|U(m'') \cap E \right|,
\end{equation} 
and return the following solution for $\cI$.  
\begin{equation}
	\label{eq:S1}
	S = \{v \in V~|~\forall i \in [2 \cdot m']: v_i \in U(m')\}. 
\end{equation} 

\begin{myclaim}
	\label{clm:Sfeasible}
	$S$ is a feasible solution for $\cI$.
\end{myclaim}

\begin{proof}
	We note that $S \subseteq V$ by \eqref{eq:S1}. Also, 
    \begin{equation}
		\label{eq:S2}
		|S| = \frac{\left| U(m') \setminus E\right|}{2 \cdot m'} \leq  \frac{|U(m')|}{2 \cdot m'}  \leq \frac{h_{m'}}{2 \cdot m'} = \frac{2m' \cdot k+m'}{2 \cdot m'} = k+\frac{1}{2}. 
	\end{equation} 
 The first equality holds by Definition~\ref{def:m-reduced} and by \eqref{eq:S1}. The second inequality holds since $\cA$ returns a feasible solution for the {\URCP} instance $\cj_m$.
 As $|S| \in \N$, we have that $|S| \leq k$.
 Therefore $S$ is a feasible solution for $\cI$. 
\qed\end{proof}

\begin{myclaim}
	\label{clm:Sprofit}
	$S$ induces in $G$ a subgraph with at least $m^*$ edges. 
\end{myclaim}

\begin{proof}
	Observe that
	\begin{equation*}
		\left|\{\{u,v\} \in E~|~ u,v \in S\}\right| \geq  |U(m') \cap E| \geq |U(m^*) \cap E| \geq m^*.  
	\end{equation*} The first inequality holds by Definition~\ref{def:m-reduced} and by \eqref{eq:S1}. The second inequality holds by \eqref{eq:m'}. The last inequality follows from Claim~\ref{clm:m*}. 
\qed
\end{proof} 
By Claims~\ref{clm:Sfeasible} and~\ref{clm:Sprofit}, we have that $S$ is an optimal solution for the denseset $k$-subgraph instance $\cI$.
The running time, for each $m \in \{0,1,\ldots, \binom{k}{2}\}$ is $f(\frac{1}{\eps_m}) \cdot \textsf{poly}(|\cI|)$, where $f$ is a computable function promised to exist since $\cA$ is an EPTAS for {\URCP}. Thus, as $m \leq k^2$ and $\eps_m \geq \frac{1}{2h_m} \geq \frac{1}{2\cdot \binom{k}{2} \cdot k+\binom{k}{2}}$, it holds that $\frac{1}{\eps_{m}} = O(k^3)$ for any $m \in \{0,1,\ldots, \binom{k}{2}\}$. Hence, the running time of the above algorithm, which finds an optimal solution for $\cI$, is $g(k) \cdot \textsf{poly}(|\cI|)$, where $g(k) = O(k^2 \cdot f(k^3))$.  \qed
\end{proof}

%% file: bounded.tex
\section{Hardness of Rule Caching with Bounded Vertex Degrees}
\label{sec:inOut}
In this section, we show that {\RCP} remains just as hard if the in-degree and out-degree of each vertex is bounded by only $2$, which gives the proof of Theorem~\ref{thm:2thm}. Specifically, given an {\RCP} instance, we construct an {\em augmented} {\RCP} instance where each vertex is transformed into a gadget with polynomially many vertices. The crucial attribute of the augmented instance is that it preserves the set of solutions of the original instance, up to the addition of the new vertices of each gadget.  

\subsection{The Construction}

The gadget  constructed for each vertex $x$ in the original instance contains two binary trees: the {\em in-tree} and the {\em out-tree}. Both binary trees contain $x$ as the root. However, in the in-tree the direction of the edges is towards the root, and in the out-tree the direction of the edges is from the root towards the leaves. The leaves of both trees represent the entire set of vertices of the instance. For each pair of vertices $x,y$ in the original instance such that $(x,y)$ is an edge, in the augmented instance there is an edge from the leaf of $y$ in the out-tree of $x$ to the leaf of $x$ in the in-tree of $y$. Finally, we add edges from the leaves of the out-tree of every vertex $x$ to the leaves of the in-tree of $x$, making the two trees a strongly connected component in the augmented graph. In the following, we give the formal definitions of the above construction.

For the remainder of this section, let $I = (G,p,k)$, where $G = (V,E)$, be an {\RCP} instance, and let $n = |V|$ be the number of vertices of the instance.  Also, let $m = \argmin\{2^c~| c \in \mathbb{N}, 2^c \geq n\}$ be the smallest power of $2$ at least as large as $n$; note that $m \leq 2 \cdot n$. Define the {\em augmented instance} $A(I) = (G_I,w_I,k_I)$ as follows. For each $t \in \mathbb{N}$, let $[t] = \{1,2,\ldots,t\}$.


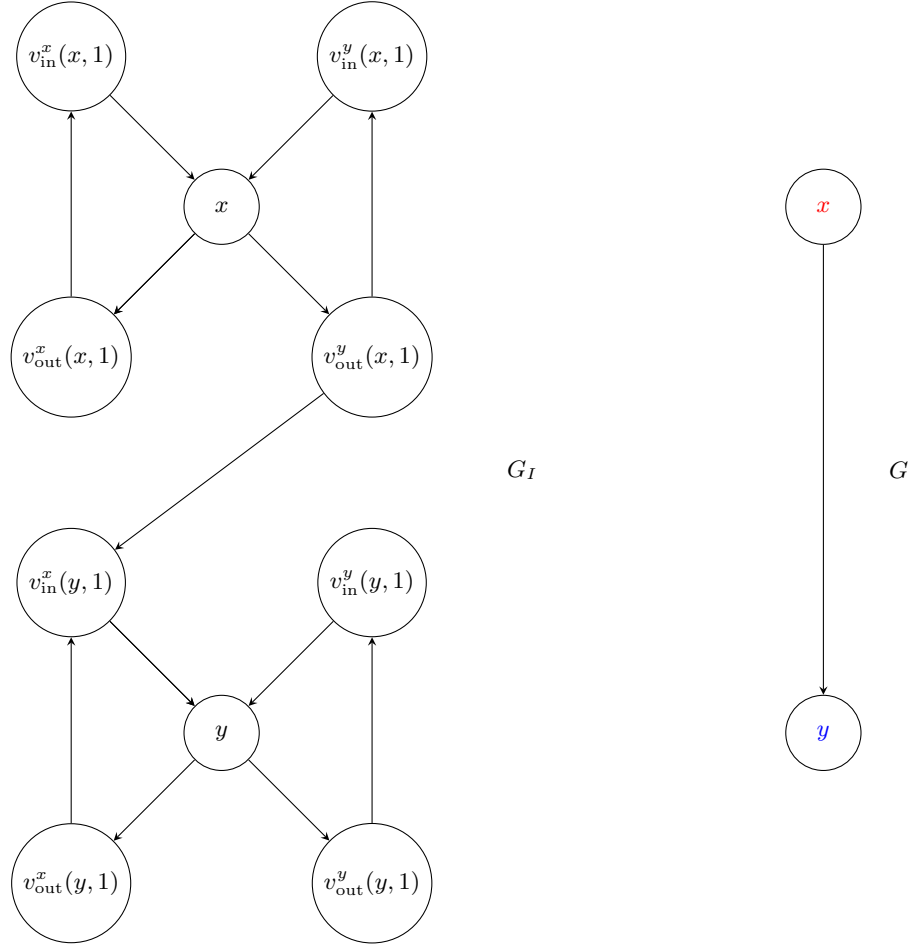
\begin{figure}
	\centering
	
	\begin{tikzpicture}
		[
		every node/.style = {circle, draw, minimum size = 1cm},
		->, >=stealth 
		]

		
		\node [rectangle, draw=none] at (4,-3.5) {$G_I$};
		
		\node [rectangle, draw=none] at (9,-3.5) {$G$};
		
		\node (X) at (8,0) {\textcolor{red}{\textbf{$x$}}};
		
		\node (Y) at (8,-7) {\textcolor{blue}{\textbf{$y$}}};

		\node (x) at (0,0) {$x$};
		\node (Oxx) at (-2,-2) {$v^x_{\text{out}}(x,1)$};
		\node (Oxy) at (2,-2) {$v^y_{\text{out}}(x,1)$};
		
		\node (Ixx) at (-2,2) {$v^x_{\text{in}}(x,1)$};
		\node (Ixy) at (2,2) {$v^y_{\text{in}}(x,1)$};
		\draw (x) edge (Oxx);
		\draw (x) edge (Oxy);
		\draw (Ixx) edge (x);
		\draw (Ixy) edge (x);
		\draw (x) edge (Oxx);
		\draw (Oxx) edge (Ixx);
		\draw (Oxy) edge (Ixy);

		\node (y) at (0,-7) {$y$};
		\node (Oyx) at (-2,-5) {$v^x_{\text{in}}(y,1)$};
		\node (Oyy) at (2,-5) {$v^y_{\text{in}}(y,1)$};
		
		\node (Iyx) at (-2,-9) {$v^x_{\text{out}}(y,1)$};
		\node (Iyy) at (2,-9) {$v^y_{\text{out}}(y,1)$};
		
		\draw  (Oyx) edge (y);
		\draw (Oyy) edge (y);
		\draw (y) edge (Iyx);
		\draw (y) edge (Iyy);
		\draw (Oyx) edge (y);
		\draw (Iyx) edge (Oyx);
		\draw (Iyy) edge (Oyy);
		
		\draw (X) edge (Y);

		\draw (Oxy) edge (Oyx);
		
	\end{tikzpicture}
	\caption{An illustration of  the graph $G_I$ constructed for $G = \left(V = \{x,y\},E = \{(x,y)\}\right)$.} \label{fig:GI}
\end{figure}

\noindent {\bf Vertices of $G_I$:} 
For each $x,y \in V$, define the vertices $v^y_{\text{in}}(x,\ell_0)$ and $v^y_{\text{out}}(x,\ell_0)$ to be the {\em leaf of} $y$ in the {\em in-tree}  and {\em out-tree} of $x$, respectively, where $\ell_0 = \log_2(m)$ indicates the distance from the root of the tree (see the following constructions). Assume without the loss of generality that $V = [n]$. Define additional $m-n$ {\em leaves} of the in-tree and out-tree of $x$, that are $v^i_{\text{in}}(x,\ell_0)$ and $v^i_{\text{out}}(x,\ell_0)$ for $i \in [m] \setminus [n]$. Let the {\em leaves} of the trees be
\begin{equation*}
	\label{eq:vertexGI}
	\begin{aligned}
		L_{\text{in}}(x) ={} & \left\{ v^y_{\text{in}}(x,\ell_0)~|~y \in V \right\} \cup  \left\{ v^i_{\text{in}}(x,\ell_0)~|~i \in [m] \setminus [n] \right\}\\
		L_{\text{out}}(x) ={} & \left\{ v^y_{\text{out}}(x,\ell_0)~|~y \in V \right\} \cup  \left\{ v^i_{\text{in}}(x,\ell_0)~|~i \in [m] \setminus [n] \right\}.
	\end{aligned}
\end{equation*} For each $\ell \in \left[ \log_2\left(\frac{m}{2} \right)\right]$ and $i \in \left[2^{\ell}\right]$, define the $i$-th intermediate vertex of {\em level} $\ell$ of the in-tree and out-tree of $x$ as $v^i_{\text{in}}(x,\ell)$ and $v^i_{\text{out}}(x,\ell)$, respectively. Let  
\begin{equation*}
	\begin{aligned}
		V_{\text{in}}(x) ={} & L_{\text{in}}(x) \cup \left\{v^i_{\text{in}}(x,\ell)~|~\ell \in \left[ \log_2\left(\frac{m}{2} \right)\right], i \in \left[2^{\ell}\right]\right\}\\
		V_{\text{out}}(x) ={} & L_{\text{out}}(x) \cup \left\{v^i_{\text{out}}(x,\ell)~|~\ell \in \left[ \log_2\left(\frac{m}{2} \right)\right], i \in \left[2^{\ell}\right]\right\}\\
	\end{aligned}
\end{equation*} We can finally define the vertex set of the augmented instance $A(I)$ as \begin{equation}
	\label{eq:V(I)}
	V_I = V \cup \bigcup_{x \in V} \left(V_{\text{in}}(x) \cup V_{\text{out}}(x)\right).  
\end{equation}

\noindent {\bf Edges of $G_I$:} The graph $G_I$ has three types of edges. The first  type generates the in-tree and out-tree for each vertex; the second type of edges guarantees that each gadget is a strongly connected component in $G_I$; the last type connects the gadgets analogously to $G$. Let the first type of edges of $x \in V$ be $$E_{1}(x) = 	E^1_{\text{in}}(x)  \cup 	E^2_{\text{in}}(x)  \cup 	E^1_{\text{out}}(x)  \cup 	E^2_{\text{out}}(x)$$
such that $E^1_{\text{in}} (x), E^2_{\text{in}}(x), E^1_{\text{out}}(x)$, and $E^2_{\text{out}}(x)$ are defined as follows.
\begin{equation}
	\label{eq:edges}
	\begin{aligned}
		E^1_{\text{in}}(x) ={} & \left\{\left(v^{2i}_{\text{in}}(x,\ell+1), v^i_{\text{in}}(x,\ell)\right)~\big|~\ell \in \left[ \log_2\left(\frac{m}{2} \right)-1\right], i \in \left[2^{\ell}\right]\right\}\\
		E^2_{\text{in}}(x) ={} & \left\{\left(v^{2i-1}_{\text{in}}(x,\ell+1), v^i_{\text{in}}(x,\ell)\right)~\big|~\ell \in \left[ \log_2\left(\frac{m}{2} \right)-1\right], i \in \left[2^{\ell}\right]\right\}\\
		E^1_{\text{out}}(x) ={} & \left\{\left(v^{i}_{\text{out}}(x,\ell), v^{2i}_{\text{out}}(x,\ell+1)\right)~\big|~\ell \in \left[ \log_2\left(\frac{m}{2} \right)-1\right], i \in \left[2^{\ell}\right]\right\}\\
		E^2_{\text{out}}(x) ={} & \left\{\left(v^{i}_{\text{out}}(x,\ell), v^{2i-1}_{\text{out}}(x,\ell+1)\right)~\big|~\ell \in \left[ \log_2\left(\frac{m}{2} \right)-1\right], i \in \left[2^{\ell}\right]\right\}\\
	\end{aligned}
\end{equation}
The second type of edges connects all corresponding pairs of leaves in the in-tree and out-tree for every vertex:
\begin{equation}
	\label{eq:secondType}
	E_2(x) = \left\{ \left(v^y_{\text{out}}(x,\ell_0), v^y_{\text{in}}(x,\ell_0) \right)~|~y \in V \right\} \cup  \left\{ \left(v^i_{\text{out}}(x,\ell_0), v^i_{\text{in}}(x,\ell_0) \right)~|~i \in [m-n] \right\}\\
\end{equation}

The third type of edges considers each pair of vertices $x,y$ in the original instance such that $(x,y) \in E$ and creates an edge from the leaf of $y$ in the out-tree of $x$ to the leaf of $x$ in the in-tree of $y$. 
\begin{equation}
	\label{eq:thirdType}
	E_3 = \left\{\left(v^{y}_{\text{out}}(x,\ell_0), v^x_{\text{in}}(y,\ell_0)\right)~|~(x,y) \in E \right\}
\end{equation}
Now, define the edges of the graph $G_I$ as \begin{equation}
	\label{eq:edgesI}
	E_I = E_3 \cup \bigcup_{x \in V} \left( E_1(x) \cup E_2(x) \right). 
\end{equation} The graph $G_I$ is simply $G_I = (V_I,E_I)$. 

\noindent {\bf Profit function $p_I$:} Define the profit function $p_I:V_I \rightarrow {\zzpos}$ such that for all $x \in V$ it holds that $p_I(x) = p(x)$ and for all $y \in V_I \setminus V$ it holds that $p_I(y) = 0$. 

\noindent {\bf Cardinality Bound:} Let $t = 1+2 \cdot \sum_{\ell \in \left[ \log_2\left(m \right)\right]} 2^{\ell}$ be the number of the vertices in the in-tree and the out-tree of any vertex $x \in V$, including $x$. Define the cardinality bound of $A(I)$ as $k_I = k \cdot t$.

\subsection{Analysis}
In the following we give some elementary properties of the construction. 

\begin{observation}
	\label{lem:pathInGadget}
	Let $x \in V$, $\ell_1,\ell_2 \in \left[ \log_2\left(m \right)\right]$, $i_1 \in \left[2^{\ell_1}\right]$, $i_2 \in \left[2^{\ell_2}\right]$, and $s_1,s_2 \in \left\{ \text{in}, \text{out}\right\}$. Then, there is a path in $G_I$ from $v^{i_1}_{s_1}(x,\ell_1)$ to  $v^{i_2}_{s_2}(x,\ell_2)$. Moreover, there is a path from $v^{i_1}_{s_1}(x,\ell_1)$ to $x$ and a path from $x$ to $v^{i_1}_{s_1}(x,\ell_1)$. 
\end{observation}


Let the symbol $+$ denote (directed) path concatenation. For simplicity we denote a path as a sequence of vertices (if the corresponding edges exist). 

\begin{lemma}
	\label{lem:Path}
	For every $x,y \in V$ there is a path from $x$ to $y$ in $G$ if and only if  there is a path from $x$ to $y$ in $G_I$.
\end{lemma}

\begin{proof}
	Let $F = \left(x = x_1,x_2,\ldots,x_r = y\right)$ be a path from $x$ to $y$ in $G$, for some $r \in \mathbb{N}$. Let $P(x)$ be a path from $x$ to $v^{x_2}_{\text{out}}(x,\ell_0)$, whose existence is guaranteed by Observation~\ref{lem:pathInGadget}. In addition, for every $q \in \{2,3,\ldots,r-1\}$ let $P(x_q)$ be a path from $v^{x_{q-1}}_{\text{in}}(x_q,\ell_0)$ to $v^{x_{q+1}}_{\text{out}}(x_q,\ell_0)$, which exists by Observation~\ref{lem:pathInGadget}. Furthermore, consider a path $P(y)$ from $v^{x_{r-1}}_{\text{in}}(y,\ell_0)$ to $y$, which exists by Observation~\ref{lem:pathInGadget}. Since $F$ is a path in $G$, by \eqref{eq:thirdType} and \eqref{eq:edgesI}, for every $q \in [r-1]$ there is an edge $\left( v^{x_{q+1}}_{\text{out}}(x_q,\ell_0),v^{x_{q}}_{\text{in}}(x_{q+1},\ell_0) \right) \in E_I$.  Thus, the following is a path from $x$ to $y$ in $G_I$:
	$$P(x)+P(x_2)+P(x_3)+\ldots+P(x_{r-1})+P(y).$$
	
	For the second direction, let $D$ be a path from $x$ to $y$ in $G_I$.  By \eqref{eq:edges}, \eqref{eq:secondType}, \eqref{eq:thirdType}, and \eqref{eq:edgesI}, there are vertices $x = y_1,y_2,\ldots, y_d = y \in V$ and paths $D(x),D(y_2),\ldots,D(y_{r-1}), D(y)$ such that the following holds. 
	\begin{enumerate}
		\item $D(x)$ is a path from $x$ to $v^{y_2}_{\text{out}}(x,\ell_0)$.
		\item For every $q \in \{2,3,\ldots,d-1\}$ it holds that $D(y_q)$ is a path from $v^{y_{q-1}}_{\text{in}}(y_q,\ell_0)$ to $v^{y_{q+1}}_{\text{out}}(y_q,\ell_0)$.
		\item $D(y)$ is a path from $v^{y_{d-1}}_{\text{in}}(y,\ell_0)$ to $y$.
		\item $D = D(y_1)+D(y_2)+\ldots+D(y_d)$. 
	\end{enumerate} By the above description of the path $D$, for every $q \in [d-1]$, it holds that $G_I$ contains the edge $\left( v^{y_{q+1}}_{\text{out}}(y_q,\ell_0),v^{y_{q}}_{\text{in}}(y_{q+1},\ell_0) \right) \in E_I$. Therefore, by \eqref{eq:thirdType} it holds that $(y_q,y_{q+1}) \in E$ for all $q \in [d-1]$. Hence, it follows that $\left(x = y_1,y_2,\ldots, y_d = y\right)$ is a path from $x$ to $y$ in $G$. \qed
\end{proof}

In the following, we discuss the complexity of the construction. 
\begin{lemma}
	\label{lem:complexityTress}
	$\left|V_I \right| \leq O\left(\left|V^2\right|\right)$. 
\end{lemma}

\begin{proof}
	Observe that $t$, the size of the gadget of each vertex, satisfies: 
	\begin{equation}
		\label{eq:tFactor}
		t = 1+2 \cdot \sum_{\ell \in \left[ \log_2\left(m \right)\right]} 2^{\ell} = O(m) = O(n) = O\left(|V| \right). 
	\end{equation} For each $x \in V$, the graph $G_I$ contains $t$ vertices in the in-tree and out-tree of $x$. Thus, 
	\begin{equation*}
		\left|V_I \right| = \left| V \cup \bigcup_{x \in V} \left(V_{\text{out}}(x) \cup V_{\text{in}}(x) \right)\right| = |V| \cdot t = O\left(\left|V^2\right|\right).
	\end{equation*} The first equality follows from \eqref{eq:vertexGI}. The last equality follows from \eqref{eq:tFactor}. \qed
\end{proof} By Lemma~\ref{lem:complexityTress}, it follows that the construction of $A(I)$ given $I$ can be computed in polynomial time in $|I|$ $-$ the encoding size of $I$. 

	
	The next results show the equivalence of the instances $A(I),I$ in terms of the quality of their solutions sets. Let $S \subseteq V$ be a solution for $I$. Define \begin{equation}
		\label{eq:SolA}
		S_I(S) = S \cup \bigcup_{x \in S} \left(V_{\text{out}}(x) \cup V_{\text{in}}(x) \right). 
	\end{equation} 
	\begin{lemma}
		\label{lem:SIsol}
		For every solution $S \subseteq V$ for $I$ it holds that $S_I(S)$ is a solution for $A(I)$ of profit $p(S)$.
	\end{lemma}
	\begin{proof}
		Let $S_I = S_I(S)$ for simplicity. Clearly, by the definition of the profit function $p_I$:
		$$p_I(S_I) = p_I(S) = p(S).$$
		To conclude, we show that $S_I$ is a solution for $A(I)$.
		%
		Let $v \in S_I$ and let $u \in V_I$ such that $(u,v) \in E_I$; we show that $u \in S_I$ by considering the following cases. If there is $x \in V$ such that $u,v \in V_{\text{out}}(x) \cup V_{\text{in}}(x) \cup \{x\}$, then by \eqref{eq:SolA} it holds that $u \in S_I$. Otherwise, by \eqref{eq:edgesI} and \eqref{eq:thirdType}, there are $x,y \in V, x \neq y$, such that $u = v^{y}_{\text{out}}(x,\ell_0)$ and  $v = v^x_{\text{in}}(y,\ell_0)$; therefore, by \eqref{eq:thirdType} it follows that $(x,y) \in E$. Since $y \in S$ by \eqref{eq:SolA} and since $S$ is a solution for $I$, it holds that $x \in S$; thus, we conclude that $V_{\text{out}}(x) \cup V_{\text{in}}(x) \subseteq S_I \cup \{x\}$ by \eqref{eq:SolA} and in particular it implies that $u \in S_I$. Finally, since the number of vertices in the in-tree and out-tree of each vertex is exactly $t$, we get: 
		$$|S_I| = \left| S \cup \bigcup_{x \in S} \left(V_{\text{out}}(x) \cup V_{\text{in}}(x) \right) \right| = t \cdot |S| \leq t \cdot k = k_I.$$
		The inequality holds since $S$ is a solution for $I$. \qed
	\end{proof}

	For the second direction of the proof, let $D_I$ be a solution for $A(I)$. Define 
	\begin{equation}
		\label{eq:Dsol}
		D(D_I) = \left\{x \in V~|~\left(V_{\text{out}}(x) \cup V_{\text{in}}(x) \cup \{x\} \right) \cap D_I \neq \emptyset \right\}. 
	\end{equation}

	To conclude, we use the following auxiliary claim.
	\begin{lemma}
		\label{clm:SIsol2}
		For every solution $D_I \subseteq V_I$ for $I$ it holds that $D(D_I)$ is a solution for $I$ of profit $p_I(D_I)$. 
	\end{lemma}
	\begin{proof}
		Let $D = D(D_I)$ for simplicity. Then, by the definition of the profit function $p_I$ it holds that 
		$$p(D) = p_I(D) = p_I(D_I).$$
		We show that $D$ satisfies the precedence constraints. Let $y \in D$ and let $x \in V$ such that $(x,y) \in E$; we show that $x \in D$. By Lemma~\ref{lem:Path}, since $(x,y) \in E$ there is a path from $x$ to $y$ in $G_I$. In addition, since $y \in D$ it holds by \eqref{eq:Dsol} that $V_{\text{out}}(y) \cup V_{\text{in}}(y) \cap D_I \neq \emptyset$; since $V_{\text{out}}(y) \cup V_{\text{in}}(y) \cup \{y\}$ is a strongly connected component in $G_I$ by Observation~\ref{lem:pathInGadget}, it follows that $V_{\text{out}}(y) \cup V_{\text{in}}(y) \cup \{y\} \subseteq D_I$; in particular, $y \in D_I$. Therefore, to satisfy the precedence constraints of $y$ in the solution $D_I$, it holds that $x \in D_I$. This implies that $x \in D$ by \eqref{eq:Dsol} as required. Finally, it holds that 
		$$|D| = \left| \left\{x \in V~|~\left(V_{\text{out}}(x) \cup V_{\text{in}}(x) \cup \{x\} \right) \cap D_I \neq \emptyset \right\}  \right| = \frac{|D_I|}{t} \leq \frac{k_I}{t} = \frac{t \cdot k}{t} = k.$$ The second equality follows since $V_{\text{out}}(y) \cup V_{\text{in}}(y) \cup \{y\}$ is a strongly connected component in $G_I$ by Observation~\ref{lem:pathInGadget}, for all $y \in V$; thus, the number of vertices in $D_I$ taken from the in-tree and out-tree of each $y \in D_I \cap V$ is exactly $t$. The inequality follows since $D_I$ is a solution for $A(I)$. \qed
	\end{proof}
	
	We can finally prove Theorem~\ref{thm:2thm} by applying the above reduction to the augmented instance. 
	
	\noindent{\bf Proof of Theorem~\ref{thm:2thm}:} Let $\rho \geq 1$ and let $\cA$ be a $\rho$-approximation algorithm for {\RCP} instances with in-degrees and out-degrees bounded by $2$. We give the following algorithm $\cB$ for general {\RCP} instances. Let $I = (G,p,k), G = (V,E)$ be an {\RCP} instance. Algorithm $\cB$ given the input $I$ is defined as follows, based on $\cA$.
	
	\begin{enumerate}
		\item Compute the augmented instance $A(I)$.
		\item Compute a solution $S_I$ for $A(I)$ using $\cA$.
		\item Return $D = D(S_I) =  \left\{x \in V~|~\left(V_{\text{out}}(x) \cup V_{\text{in}}(x) \cup \{x\} \right) \cap S_I \neq \emptyset \right\}. $ 
	\end{enumerate} Clearly, the running time of $\cB$ is polynomial in $|I|$, the encoding size of $I$; the first step takes $|I|^{O(1)}$ time by Lemma~\ref{lem:complexityTress} and the second step takes $|I|^{O(1)}$ time by the definition of $\cA$. In addition, by the definition of $\cA$ we have $$p(D) = p_I(S_I) \geq \frac{\OPT(A(I))}{\rho} \geq \frac{\OPT(I)}{\rho}.$$
	The first equality follows from Lemma~\ref{clm:SIsol2}. The first inequality holds since $\cA$ is a $\rho$-approximation algorithm for {\RCP} instances with in-degrees and out-degrees bounded by $2$, and $A(I)$ is such an {\RCP} instance. The last inequality follows from  Lemma~\ref{lem:SIsol}. \qed

%% file: missingProofs.tex
\section{Missing proofs}
\label{app:proofs}
\begin{proof}[of \cref{lem:leftover}]
	In \cref{alg:nextfit} we iterate over the children of $a$ in $T_a$, and for each such a child $u_s$ we add 
	$u_s$ and all its descendants in $T_a$ to one of the subsets $Q_m$. Thus, if $Q_m$ is one of the subsets returned by Procedure NextFit (\cref{alg:nextfit}) when it is called for $a$, then $u_s$ and all its descendants in $T_a$ are anchored at $a$; 
    otherwise, $u_s$ and all its descendants are leftover vertices of $a$. \qed
\end{proof}

\begin{proof}[of \cref{lem:at_most_one_anchor}]
	To prove the lemma, it suffices to show that a vertex cannot be anchored at more than one anchor.
	Fix an iteration $t$. Consider two anchors $a$ and $a'$ added in iteration $t$, namely $t(a)=t(a')=t$. Note that $a$ is neither a descendant nor an ancestor of $a'$, and thus the set of vertices anchored at $a$, which is contained in $\des{T[V_{t}]}{a}$ is disjoint from the set of vertices anchored at $a'$, which is contained in $\des{T[V_{t}]}{a'}$.
	To complete the proof, we note that for any two anchors $a$ and $a'$ such that $t(a)<t(a')$, the set of vertices that are anchored at $a$ are disjoint from the set of vertices that are anchored at $a'$. This holds since by \cref{lem:leftover} all the leaves of $T$ that are descendants of the vertices anchored at $a$ are covered at some iteration $t\le t(a)$; thus, the vertices anchored at $a$ are not in $V_{t(a)+1}$. \qed
\end{proof}